\documentclass{article}
\usepackage[numbers]{natbib} 
\usepackage{arxiv}

\renewcommand{\headeright}{}
\renewcommand{\undertitle}{}
\usepackage{graphics}

\usepackage{graphicx}
\usepackage{microtype}

\usepackage{lineno}
\usepackage{amsmath,amsfonts,amssymb}
\usepackage{ntheorem}
\newenvironment{prop}[1][Proposition]{\par\noindent\textbf{#1.} \itshape}{\par}
\theoremstyle{plain}
\theorembodyfont{\normalfont}
\newtheorem*{proof}{Proof}
\theoremstyle{plain}
\theorembodyfont{\normalfont}
\newtheorem{definition}{Definition}
\usepackage{etoolbox} 
\usepackage{siunitx}

\newcommand*\linenomathpatch[1]{%
  \cspreto{#1}{\linenomath}%
  \cspreto{#1*}{\linenomath}%
  \csappto{end#1}{\endlinenomath}%
  \csappto{end#1*}{\endlinenomath}%
}

\linenomathpatch{equation}
\linenomathpatch{gather}
\linenomathpatch{multline}
\linenomathpatch{align}
\linenomathpatch{alignat}
\linenomathpatch{flalign}

\usepackage{bm}
\usepackage[position=top,labelfont=normalfont,textfont=normalfont,singlelinecheck=off,justification=raggedright]{subfig}
\usepackage{floatrow}
\usepackage[dvipsnames]{xcolor}
\usepackage{setspace}
\usepackage{enumerate,etaremune}
\usepackage{booktabs}
\usepackage{tabularx,multirow}
\usepackage{framed}
\usepackage{hhline}
\usepackage{url}
\usepackage[unicode,bookmarks=false]{hyperref}
\hypersetup{
  colorlinks,
  citecolor=Blue,linkcolor=Blue,urlcolor=Blue
}
\usepackage[toc,page]{appendix}
\usepackage[T1]{fontenc}

\usepackage{tikz}
\usetikzlibrary{shapes.geometric}
\usetikzlibrary{shapes,arrows}
\usetikzlibrary{plotmarks}
\usetikzlibrary{calc}
\usetikzlibrary{
  arrows,
  calc,
  fit,
  patterns,
  plotmarks,
  shapes.geometric,
  shapes.misc,
  shapes.symbols,
  shapes.arrows,
  shapes.callouts,
  shapes.multipart,
  shapes.gates.logic.US,
  shapes.gates.logic.IEC,
  er,
  automata,
  backgrounds,
  chains,
  topaths,
  trees,
  petri,
  mindmap,
  matrix,
  calendar,
  folding,
  fadings,
  through,
  positioning,
  scopes,
  decorations.fractals,
  decorations.shapes,
  decorations.text,
  decorations.pathmorphing,
  decorations.pathreplacing,
  decorations.footprints,
  decorations.markings,
  shadows}
\usetikzlibrary{tikzmark}

\usepackage{algorithm}
\usepackage{algorithmicx,algpseudocode}
\usepackage{cases}
\algdef{SE}[DOWHILE]{Do}{doWhile}{\algorithmicdo}[1]{\algorithmicwhile\ #1}

\def\Figref#1{Figure~\ref{#1}}

\def\Secref#1{Section~\ref{#1}}

\def\eqref#1{(\ref{#1})}
\def\Eqref#1{(\ref{#1})}

\def\Algref#1{Algorithm~\ref{#1}}

\def\Appref#1{Appendix~\ref{#1}}

\def\1{\bm{1}}

\def\vb{{\bm{b}}}
\def\vc{{\bm{c}}}

\def\vf{{\bm{f}}}
\def\vg{{\bm{g}}}

\def\vp{{\bm{p}}}

\def\vu{{\bm{u}}}
\def\vv{{\bm{v}}}
\def\vw{{\bm{w}}}
\def\vx{{\bm{x}}}

\def\vz{{\bm{z}}}

\def\mA{{\bm{A}}}
\def\mB{{\bm{B}}}

\def\mD{{\bm{D}}}

\def\mF{{\bm{F}}}
\def\mG{{\bm{G}}}

\def\mI{{\bm{I}}}

\def\mM{{\bm{M}}}

\def\mP{{\bm{P}}}

\def\mV{{\bm{V}}}

\def\mX{{\bm{X}}}

\DeclareMathAlphabet{\mathsfit}{\encodingdefault}{\sfdefault}{m}{sl}
\SetMathAlphabet{\mathsfit}{bold}{\encodingdefault}{\sfdefault}{bx}{n}

\def\gB{{\mathcal{B}}}
\def\gC{{\mathcal{C}}}

\def\gN{{\mathcal{N}}}
\def\gO{{\mathcal{O}}}

\def\emA{{A}}

\DeclareMathOperator*{\argmin}{arg\,min}

\newcommand{\eg}{{\it e.g.}}
\newcommand{\ie}{{\it i.e.}}

\newcommand{\etal}{{\it et al.}}
\newcommand{\tensor}[1]{\bm{#1}}

\newsavebox{\dotbox}

\newcolumntype{L}[1]{>{\raggedright\let\newline\\arraybackslash\hspace{0pt}}m{#1}}
\newcolumntype{C}[1]{>{\centering\let\newline\\arraybackslash\hspace{0pt}}m{#1}}
\newcolumntype{R}[1]{>{\raggedleft\let\newline\\arraybackslash\hspace{0pt}}m{#1}}

\allowdisplaybreaks

\usepackage{xcolor}
\definecolor{darkred}{rgb}{0.55, 0.0, 0.0}
\definecolor{skyblue}{rgb}{0.53, 0.81, 0.92}

\newcommand{\revision}[1]{\textcolor{black}{#1}}

\newcommand{\ours}{UMLS-MPM}

\def\argmin{\mathop{\rm argmin}}

\usepackage{hyperref}
\usepackage{soul}
\usepackage{float}
\usepackage{authblk} 

\begin{document}
\title{Unstructured Moving Least Squares Material Point Methods: A Stable Kernel Approach With Continuous Gradient Reconstruction on General Unstructured Tessellations}

\author[1]{Yadi Cao}
\author[2,1]{Yidong Zhao}
\author[3,1]{Minchen Li}
\author[4]{Yin Yang}
\author[2]{Jinhyun Choo}
\author[1]{Demetri Terzopoulos}
\author[1]{Chenfanfu Jiang}

\affil[1]{University of California, Los Angeles, Los Angeles}
\affil[2]{Korea Advanced Institute of Science \& Technology}
\affil[3]{Carnegie Mellon University}
\affil[4]{University of Utah}

\date{}

%
%
\maketitle
\begin{abstract}
The Material Point Method (MPM) is a hybrid Eulerian Lagrangian simulation technique for solid mechanics with significant deformation.
Structured background grids are commonly employed in the standard MPM, but they may give rise to several accuracy problems in handling complex geometries.
When using (2D) unstructured triangular or (3D) tetrahedral background elements, however, significant challenges arise (\eg, cell-crossing error).
Substantial numerical errors develop due to the inherent $\mathcal{C}^0$ continuity property of the interpolation function, which causes discontinuous gradients across element boundaries.
Prior efforts in constructing $\mathcal{C}^1$ continuous interpolation functions have either not been adapted for unstructured grids or have only been applied to 2D triangular meshes. 
In this study, an Unstructured Moving Least Squares MPM (UMLS-MPM) is introduced to accommodate 2D and 3D simplex tessellation. 
The central idea is to incorporate a diminishing function into the sample weights of the MLS kernel, ensuring an analytically continuous velocity gradient estimation. 
Numerical analyses confirm the method's capability in mitigating cell crossing inaccuracies and realizing expected convergence.
\end{abstract}
\section{Introduction}
\label{sec:intro}

The Material Point Method (MPM)~\cite{sulsky1995application} was introduced to solid mechanics as an extension of both the Fluid-Implicit Particle (FLIP) method~\cite{brackbill1988flip} and the Particle-in-Cell (PIC) method~\cite{harlow1962particle}. The MPM is a hybrid Eulerian-Lagrangian method, often referred to as a particle-grid method that retains and monitors all physical attributes on a collection of particles. A background grid serves in solving the governing equations. Both Eulerian and Lagrangian descriptions are incorporated in the MPM to overcome the numerical challenges stemming from nonlinear convective terms inherent in a strictly Eulerian approach, while avoiding significant grid distortions typically found in purely Lagrangian methods. The efficacy of the method has been demonstrated in problems concerning extreme deformation of solid materials, such as biological soft tissues~\cite{ionescu2006simulation,guilkey2006computational}, explosive materials~\cite{guilkey2007eulerian,ma2009simulation}, sand~\cite{homel2014simulation,klar2016drucker,tampubolon2017multi}, and snow~\cite{stomakhin2013material,gaume2018dynamic,gaume2019investigating}.

Based on the specific Lagrangian formulations, the
MPM is categorized into total Lagrangian~\cite{de2020total,de2021modelling,de2022total,de2022mesh} and updated Lagrangian~\cite{pretti2023conservation} variants, in which equations are formulated in different reference configurations. In the total Lagrangian MPM, numerical dissipation errors or artificial fractures are not observed; 
however, challenges arise due to mesh distortions as the connectivity is preserved in a manner similar to the Finite Element Method (FEM). Conversely, the updated Lagrangian MPM has been found to exhibit greater robustness, particularly in dealing with demanding scenarios such as impacts and shocks~\cite{zhang2006explicit,huang2010contact,homel2014simulation,cao2022efficient}, failures and cracks in both single-phase and multi-phase materials~\cite{chen2002evaluation,zhang2009material,tampubolon2017multi}, and contact mechanics~\cite{lian2011coupling,chen2015improved,homel2017field,homel2018fracture,cheon2018efficient,de2021modelling,guilkey2021hybrid,nakamura2021particle}.

Despite its numerous successes, the updated Lagrangian MPM mainly adopts a uniformly-structured background grid that aligns with the axes of the global Cartesian coordinate system, using (2D) quadrilaterals or (3D) hexahedra for spatial discretization. 
When boundaries involve complex geometry, however, the aforementioned approach may introduce significant challenges in conformally discretizing the space. 
Remarkably, many engineering problems, such as those in mechanical and geotechnical engineering~\cite{fern2019material}, involve complex boundary geometry.
Hence, some researchers~\cite{wikeckowski2004material,beuth2011solution,jassim2013two,wang2021efficient} have proposed using unstructured (2D) triangles or (3D) tetrahedra for discretization, which provides substantial flexibility in the presence of geometrically complex boundaries. 

Unfortunately, most of the existing approaches using unstructured triangular or tetrahedral elements adopt a piecewise linear ($\gC^0$) basis function~\cite{wikeckowski2004material,beuth2011solution,jassim2013two,wang2021efficient} whose gradient is discontinuous along element boundaries. 
In this case, when particles move from one element to another (\ie, crossing element boundaries), a significant error arises---the so-called cell-crossing error~\cite{bardenhagen2004generalized}. 
Because the function gradient becomes discontinuous along element boundaries, the cell-crossing error leads to severe stress oscillations, causing significant numerical errors. 

Several approaches have been proposed for circumventing the cell-crossing error, including the generalized interpolation material point (GIMP) method \cite{bardenhagen2004generalized,charlton2017igimp}, the dual domain MPM (DDMPM)~\cite{zhang2011material}, the use of high-order basis functions such as B-splines~\cite{steffen2008analysis,gan2018enhancement}, and approaches based on moving least squares (MLS) basis functions~\cite{hu2018moving,tran2019moving}. 
Unfortunately, they are either limited to structured quadrilaterals/hexahedra or are only applicable to 2D cases using triangles~\cite{de2021extension}.
This leaves the cell-crossing error as an unsolved challenge when using unstructured tessellations in both the 2D and 3D MPM.

The objective of this study is to address the aforementioned cell-crossing challenge for general unstructured meshes in both 2D and 3D. 
The proposed approach is built upon a new MLS reconstruction process that is suitable for general unstructured discretization.
By incorporating a diminishing function into the sample weights of the MLS kernel, an analytically continuous function gradient is achieved, which efficiently eliminates the cell-crossing error.
A new MLS kernel function is derived that can be straightforwardly implemented into an existing MPM framework.

The remainder of this paper is structured as follows: \Secref{sec:governing_eqs} introduces the general governing equations of the MPM and the details of a typical explicit MPM process. The Moving Least Squares (MLS) approximation and the MLS-MPM method are discussed in \Secref{sec:interpws:mls}. A seemingly straightforward yet inherently flawed extension of the MLS-MPM to unstructured meshes, along with the associated cell-crossing challenge, is presented in \Secref{sec:interpws:mls:unstructured}. \Secref{sec:interpws:mls:unstructured:remedy} develops a solution to this challenge, accompanied by an in-depth analysis and kernel reconstruction for representative unstructured meshes. Numerical results affirming the efficacy of the proposed method are reported in \Secref{sec:experiments}. The paper concludes in \Secref{sec:conclusion} with reflections and recommendations for future work.

\section{Methodology}
\label{sec:method}

\subsection{Governing Equations}
\label{sec:governing_eqs}

Following standard continuum mechanics \cite{bonet2008nonlinear}, consider the mapping $\vx = \bm{\phi}(\mX, t)$, which maps points from the (reference) material configuration, represented by $\mX$, to their corresponding locations in the (current) spatial configuration, represented by $\vx$. In this framework, velocity is defined in two different but equivalent manners. On the one hand, $\mV(\mX,t) = \frac{\partial \vx}{\partial t}(\mX, t)$ defines the Lagrangian velocity in the material configuration. On the other hand, the Eulerian velocity in the spatial configuration, is denoted by $\vv(\vx,t) = \mV(\bm{\phi}^{-1}(\vx,t),t)$. Furthermore, the deformation experienced by the material points is quantified using the deformation gradient, given by $\mF(\mX,t) = \frac{\partial \vx}{\partial \mX}(\mX, t)$. The determinant of this gradient, represented by $J$, is also crucial as it provides insights into volumetric changes associated with the deformation process.

Given these definitions, the conservation equations for mass and momentum (neglecting external forces) are \cite{zhang2016material,bonet2008nonlinear}
\begin{equation}
    \begin{aligned}
        \rho J                 & = \rho_0,                   \\
        \rho \frac{D \bm{v}}{D t} & = \nabla \cdot \bm{\sigma},
    \end{aligned}
\end{equation}
where $\rho$ represents the density,  $D/Dt$ is the material derivative, and
\begin{equation}
    \bm{\sigma} = \frac{1}{J} \bm{P} \bm{F}^T.
\end{equation}
is the Cauchy stress tensor, which is related to the first Piola-Kirchhoff stress $\bm{P}=\frac{\partial \Psi}{\partial \bm{F}}$, where $\Psi$ denotes the strain energy density.
The evolution of the deformation gradient is given by
\begin{equation}
    \label{eq:mpm:deformation_gradient}
    \dot{\mF} = (\nabla \bm{v}) \mF.
\end{equation}

Consider a domain represented by $\Omega$. Boundaries on which the displacement is known, represented as $\partial \Omega_u$, are governed by the Dirichlet boundary condition
\begin{equation}
    x_k(\vx,t) = \bar{x}_k(\vx,t), \quad \forall \vx \in \partial \Omega_u, 
\end{equation}
where $\bar{x}_k$ denotes the predetermined displacement for component $k$. Boundaries on which the tractions (forces per unit area) are predefined, represented as $\partial \Omega_{\tau}$, adhere to the Neumann boundary condition
\begin{equation}
    \sigma_{kl}(\vx,t) n_l = \bar{\tau}_k(\vx,t), \quad \forall \vx \in \partial \Omega_{\tau},
\end{equation}
where $\bar{\tau}_k$ is the prescribed traction for component $k$, and $\sigma_{kl}(\vx,t) n_l$ represents the traction inferred from the stress tensor $\sigma_{kl}$ acting in the direction of the outward unit normal vector $n_l$. For ease of reference and notational clarity in our framework, the subscripts $k$ and $l$ refer to components $k$ and $l$ of any given vector or tensor.

To solve the conservation equations for mass and momentum within the MPM framework, one often turns to the weak form. Specifically, a continuous test function $\phi$, which vanishes on $\partial \Omega_u$, is employed. Then, both sides of the equation are multiplied by $\phi$ and integrated over the domain $\Omega$:
\begin{equation}
    \label{eq:mpm:weak_form}
    \int_{\Omega} \phi \rho \ddot{x}_k  d\Omega = \int_{\partial \Omega_{\tau}} \phi \tau_{k} dA - \int_{\Omega} \frac{\partial \phi}{\partial \vx_l} \sigma_{kl}  d\Omega.
\end{equation}
At this juncture, integration by parts and the Gauss integration theorem are utilized, nullifying the contributions on $\partial \Omega_u$ due to the vanishing of the test function on this boundary subset.

\revision{
For clarity, in the remainder of this paper, the terms "grids" or "grid nodes" will exclusively refer to regular background grid nodes. In contrast, "mesh nodes" will denote nodes in general, unstructured meshes.
For simplicity, we do not change the common abbreviations such as Particle-To-Grid (P2G) and Grid-To-Particle (G2P).
}

In the standard implementation of the MPM, physical quantities such as mass and velocity are retained at material points and then projected onto background grid nodes for further computation. \Eqref{eq:mpm:weak_form} is discretized on these nodes by the Finite Element Method (FEM) and then solved using either implicit or explicit time integration schemes. This article focuses on the explicit symplectic Euler time integration method. While the extension to implicit methods is possible and straightforward, it would be orthogonal to the contribution of the article.

\subsection{Explicit MPM Pipeline}
\label{sec:pipeline}

The explicit MPM pipeline in each time step has four main stages: (1) the transfer of material point quantities to the background nodes, known as Particle-To-Grid (P2G), (2) the computation of the system's evolution on these background nodes, (3) the back-transfer of the evolved quantities to the material points, known as Grid-To-Particle (G2P), and (4) the execution of necessary post-processings, such as elastoplasticity return mapping and material hardening. \Algref{alg:explicit:mpm_pipeline} presents an overview of the MPM pipeline, and the main stages are elaborated below.

\begin{algorithm}[t]
    \caption{Explicit MPM}
    \begin{algorithmic}[1]
        \State Determine material point-node connectivity, calculate kernel functions $w_{p, i}$
        \State \textbf{P2G:}
        \Statex \hspace{1cm}  Nodal mass: $m_i = \sum_p \rho_p V_p w_{p, i}$
        \Statex \hspace{1cm}  Nodal momentum: $\vp_i = \sum_p \vv_p \rho_p V_p w_{p, i}$
        \Statex \hspace{1cm}  Nodal velocity: $\vv_i = \vp_i / m_i$
        \State \hspace{1cm}  Internal force: $\vf_i^{\text{int}} = -\sum_p V_p \bm{\sigma}_p \nabla w_{p, i}$
        \State \hspace{1cm}  Gravity: $\vf_i^{\text{ext}} = \sum_p w_{p, i} m_p \mathbf{g}_p$
        \State \hspace{1cm}  Nodal force: $\vf_i = \vf_i^{\text{ext}} + \vf_i^{\text{int}}$
        \State \textbf{Deformation of background nodes:}
        \Statex \hspace{1cm} Updated nodal accelerations: $\ddot{\vx}_i = \vf_i / m_i$
        \Statex \hspace{1cm} Update nodal velocities: $\tilde{\vv}_i = \vv_i + \Delta t \ddot{\vx}_i$
        \Statex \hspace{1cm} Enforce Dirichlet conditions: $\ddot{\vx}_i = 0 \text{ and } \vf_i = 0$
        \State \textbf{G2P:}
        \Statex \hspace{1cm} Update point velocities: $\vv_p^{\Delta t} = \vv_p + \Delta t \sum_i w_{p, i} \ddot{\vx}_i $
        \Statex \hspace{1cm} Update point positions: $\vx_p^{\Delta t} = \sum_i w_{p, i} \tilde{\vx}_i$
        \State Update deformation gradient: $\mF_p^{\Delta t} = \left(\mI + \sum_i (\tilde{\vx}_i-\vx_i)(\nabla w_{p, i})^T\right) \mF_p$
        \State Update point volume: $V^{\Delta t}_p = \text{det}(\mF^{\Delta t}_p) V^0_p$
        \State Update point stresses: $\bm{\sigma}_p = C(\mF_p)$
        \State \textbf{Enforce plasticity, reset background deformation, advance to next timestep}
    \end{algorithmic}
    \label{alg:explicit:mpm_pipeline}
\end{algorithm}

\paragraph{Stage 1: P2G}
In the MPM, material points are the Lagrangian particles that track the location of the continuum along with physical attributes such as mass, position, and velocity.
\revision{
To evolve the dynamics on the background grid or mesh nodes, an interpolation function—also known as a transfer kernel or simply a kernel—needs to be determined to relate the information from particles to their nearby active nodes. Generally, for a particle located at $\vx_p$ and all surrounding nodes at $\vx_1, \dots, \vx_N$, the stacked kernel values associating the two sides are:
}

\revision{
\begin{equation}
    \label{eq:kernel:define}
    \vw_p
    = [w_{p, 1}, \dots, w_{p, N}]^T
    = \vw(\vz ; \vx_p, \vx_1, \dots, \vx_N) \bigg|_{\vz = \vx_p}.
\end{equation}
}

\revision{
Specifically, we include $\vx_1, \dots, \vx_N$ in this definition because constraints, such as the partition of unity, can only be determined while considering all active neighbors. Nonetheless, the neighbors are implicitly detected by $\vx_p$; for conciseness, we omit this implicit condition in the remaining part of this paper. The kernel supports transferring information between the nodes and any location $\vz$ near $\vx_p$; in most MPM works, the kernel is evaluated at the lagged $\vx_p$ before the completion of P2G, i.e., we set $\vz = \vx_p$ and keep the kernel unchanged for both P2G and G2P.
}

\revision{
Following this, the stacked gradients of the kernels are obtained with respect to the spatial variable $\vz$, then evaluated at $\vx_p$:
}

\revision{
\begin{equation}
    \label{eq:kernel:gradient:define}
    \mG_p
    = [\vg_{p, 1}, \dots, \vg_{p, N}]^T
    = \left[\nabla_{\vz} \vw(\vz ; \vx_p) \right] \bigg|_{\vz = \vx_p}.
\end{equation}
}

In the explicit MPM framework, the lumped mass at each background node is defined as $m_i = \sum_p \rho_p V_p w_{p, i}$, where $\rho_p$ represents the density and $V_p$ the volume of each nearby particle. This definition facilitates the calculation of the background node momentum, expressed as
\begin{equation}
    \label{eq:explicit:mpm:momentum}
    m_i\ddot{\vx}_i = \vf^{\text{int}}_i + \vf^{\text{ext}}_i,
\end{equation}
where $\ddot{\vx}_i$ is the acceleration of node $i$, and $\vf^{\text{int}}_i$ and $\vf^{\text{ext}}_i$ represent the internal forces and the external forces acting on the it, respectively:
\begin{align}
    \vf^{\text{int}}_i & = -\sum_p V_p \bm{\sigma}_p \nabla w_{p, i}, \\
    \vf^{\text{ext}}_i & = \sum_p m_p w_{p, i}\vb_p + \sum_p m_p w_{p, i}\vg_p.
\end{align}
The stress tensor $\bm{\sigma}$ is determined by the deformation gradient $\mF$ through some constitutive relation, indicating how material deformation influences internal forces:
\begin{equation}
    \bm{\sigma} = C(\mF).
\end{equation}

\paragraph{Stage 2: Evolution on the Background Nodes}

Using the accelerations obtained from \Eqref{eq:explicit:mpm:momentum}, we integrate the velocities and positions of the background nodes using a symplectic Euler time integrator employed throughout this work:
\begin{align}
    \tilde{\vv}_i & = \vv_i + \Delta t \ddot{\vx}_i  \quad \text{(Velocity Update)}, \\
    \tilde{\vx}_i & = \vx_i + \Delta t \tilde{\vv}_i \quad \text{(Position Update)},
\end{align}
where the time step size $\Delta t$ is chosen based on the CFL condition~\cite{ferziger2019computational}.

\paragraph{Stage 3: G2P}

The FLIP scheme~\cite{brackbill1988flip} is utilized for all experiments discussed in \Secref{sec:experiments}. In FLIP, the particle positions and velocities are updated as follows:
\begin{align}
    x_p^{\Delta t} & = \sum_i w_{p, i}\tilde{\vx}_i,                \\
    v_p^{\Delta t} & = v_p + \Delta t \sum_i w_{p, i} \ddot{\vx}_i.
\end{align}
Subsequently, the evolution of the deformation gradient $\mF$ in \Eqref{eq:mpm:deformation_gradient} is conducted as follows:
\begin{equation}
    \mF_p^{\Delta t} = \left(\mI + \sum_i (\tilde{\vx}_i-\vx_i)(\nabla w_{p, i})^T\right) \mF_p.
\end{equation}
Given the initial deformation gradient $\mF^0 = \mI$ and initial volume $V^0_p$, particle volumes are updated as:
\begin{equation}
    V^{\Delta t}_p = \text{det}(\mF^{\Delta t}_p) V^0_p.
\end{equation}

\paragraph{Stage 4: Post-Processing and Resetting the Background Nodes}

This stage encompasses all post-processing tasks such as plasticity return mapping and material hardening \cite{simo2006computational}. In the updated Lagrangian MPM, the grid is reset to a non-deformed state at the end of each timestep. This is achieved by keeping the grid or mesh constant while zeroing all information such as velocity and acceleration.

\subsection{Transfer Kernel}
\label{sec:interpws}

In the MPM, the transfer kernel is crucial for relaying particle information to adjacent background nodes. Techniques such as the B-spline MPM~\cite{steffen2008analysis} and GIMP~\cite{bardenhagen2004generalized} use a specific compact support function to smoothly influence nearby grid nodes, whereas methods like Moving Least Squares MPM (MLS-MPM)~\cite{hu2018moving} determine the kernel implicitly, based on the proximity of nodes. However, both strategies follow a similar workflow, which involves for every particle: (1) identifying the set of nearby nodes, and (2) calculating the transfer kernel and gradient for every particle-node pair.

This section first introduces the general MLS reconstruction process and the application of MLS-MPM with a comprehensive linear polynomial basis. It is followed by a discussion on a naive extension of MLS-MPM to unstructured meshes, highlighting the steps of identifying nearby nodes and computing the transfer weights. We then delve into the desirable properties of the kernel, emphasizing why the naive extension fails to yield continuous gradient reconstructions when particles cross cell boundaries. Finally, we propose a solution addressing the issue of discontinuous gradient reconstructions and introduce \ours.


\subsubsection{Introduction to General MLS and MLS-MPM}
\label{sec:interpws:mls}
\revision{
Given the kernel definition in \Eqref{eq:kernel:define}, the Moving Least Squares (MLS) method aims to use a polynomial-based kernel to reconstruct $\hat{u}$ for some function $u$ at any location $\vz$ near a given particle location $\vx_p$. It is defined as follows:
\begin{equation}
    \label{eq:mls:general}
    \hat{u}(\vz;\vx_p) = \vp^T(\vz - \vx_p) \vc(\vx_p),
\end{equation}
where $\vp(\vz - \vx_p) = [p_0(\vz - \vx_p), p_1(\vz - \vx_p), \ldots, p_l(\vz - \vx_p)]^T$ represents the polynomial basis, $\vc(\vx_p) = [c_0(\vx_p), c_1(\vx_p), \ldots, c_l(\vx_p)]^T$ are the corresponding coefficients, and $l$ indicates the total order of the basis. The coefficients $\vc(\vx_p)$ are determined by minimizing the sum of weighted square errors between the sampled function values $u_i$ and the reconstructed values $\hat{u}_i$ at nearby node positions $\vx_i$:
\begin{equation}
    \label{eq:mls:general:minimize}
    \vc(\vx_p) = \argmin \sum_{i \in \gB_{\vx_p}} d(\vx_i - \vx_p) \, ||u_i - \vp^T(\vx_i - \vx_p) \vc(\vx_p)||^2,
\end{equation}
where $d$ is a weighting function that takes proximity as input, and $\gB_{\vx_p}$ is the set of sample points in the local region around $\vx_p$ where the weighting function is non-zero.
}

\revision{
This minimization leads to the following solution for $\vc(\vx_p)$:
\begin{equation}
    \label{eq:mls:general:sol}
    \vc(\vx_p) = \mM^{-1}(\vx_p) \mB(\vx_p) \vu,
\end{equation}
where
\begin{equation}
    \label{eq:mls:general:M}
    \begin{aligned}
        \mM(\vx_p) &= \sum_{i \in \gB_{\vx_p}} d(\vx_i - \vx_p) \vp(\vx_i - \vx_p) \vp^T(\vx_i - \vx_p) \\
        &= \mP(\vx_p) \mD(\vx_p) \mP(\vx_p)^T,
    \end{aligned}
\end{equation}
and
\begin{equation}
    \label{eq:mls:general:b}
    \mB(\vx_p) = \mP(\vx_p) \mD(\vx_p).
\end{equation}
}

\revision{
Here we use the stacked notations: $\vu = [u_1, \ldots, u_N]^T$ is the stacked sample values, $\mP(\vx_p) = [\vp(\vx_1 - \vx_p), \ldots, \vp(\vx_N - \vx_p)]$ is the stacked basis, and $\mD(\vx_p)$ is the diagonal sample weighting matrix with $\mD_{i,i}(\vx_p) = d(\vx_i - \vx_p)$.
}

\revision{
Substituting \Eqref{eq:mls:general:sol} into \Eqref{eq:mls:general}, we obtain the reconstruction:
\begin{equation}
    \label{eq:mls:general:recon}
    \begin{aligned}
        \hat{u}(\vz;\vx_p) &= \vp^T(\vz - \vx_p) \mM^{-1}(\vx_p) \mB(\vx_p) \vu \\
        &= \vw(\vz;\vx_p)^T \vu,
    \end{aligned}
\end{equation} 
where the last derivation is obtained by defining the kernel for MLS, and note that $\mM^{-1}(\vx_p)$ is symmetric:
\begin{equation}
    \label{eq:mls:general:kernel:def}
    \vw(\vz;\vx_p) = \mB^T(\vx_p) \mM^{-1}(\vx_p) \vp(\vz - \vx_p).
\end{equation}
}

\revision{
Similar to \Eqref{eq:kernel:define} and \Eqref{eq:kernel:gradient:define}, in the context of MPM, we again set the spatial variable $\vz$ to be the particle positions $\vx_p$ before completing P2G and obtain the kernel value:
\begin{equation}
    \label{eq:mls:mpm:kernel}
    \vw_p
    = \vw(\vz;\vx_p)\bigg|_{\vz = \vx_p}
    = \mB^T(\vx_p) \mM^{-1}(\vx_p) \vp(0),
\end{equation}
as well as the kernel gradient:
\begin{equation}
    \label{eq:mls:mpm:kernel:gradient}
    \mG_p
    = \left[\nabla_{\vz} \vw(\vz ; \vx_p) \right] \bigg|_{\vz = \vx_p}
    = \mB^T(\vx_p) \mM^{-1}(\vx_p) (\nabla_{\vz} \vp)(0).
\end{equation}
}

\paragraph{The Linear Polynomial Basis Case}
\revision{
A special case involves using a complete linear polynomial basis, as in MLS-MPM~\cite{hu2018moving}, where $\vp(\vz - \vx_p) = [1, (\vz - \vx_p)^T]^T$. Setting $\vz = \vx_p$ in this basis, we have:
\begin{equation}
    \label{eq:mls:mpm:linear:basis:zero}
    \begin{aligned}
        \vp(0) &= 
        \begin{bmatrix}
            1 \\
            \bm{0}_{\text{dim}}
        \end{bmatrix}, \\
        (\nabla_{\vz} \vp)(0) &= 
        \begin{bmatrix}
            \bm{0}^T_{\text{dim}} \\
            \mI_{\text{dim}, \text{dim}}
        \end{bmatrix}, \\
        [\vp(0), (\nabla_{\vz} \vp)(0)] &= \mI_{\text{dim}+1, \text{dim}+1};
    \end{aligned}
\end{equation}
}

\revision{
Substituting \Eqref{eq:mls:mpm:linear:basis:zero} into \Eqref{eq:mls:mpm:kernel} and \Eqref{eq:mls:mpm:kernel:gradient} and stacking $\vw_p$ and $\mG_p$ in a column, we obtain a compact formula for both the kernel and the gradient:
\begin{equation}
    \label{eq:mls:mpm:linear:kernel:gradient}
    \begin{aligned}
        [\vw_p, \mG_p] &= \mB^T(\vx_p) \mM^{-1}(\vx_p) [\vp(0), (\nabla_{\vz} \vp)(0)] \\
        &= \mB^T(\vx_p) \mM^{-1}(\vx_p) \mI_{\text{dim}+1, \text{dim}+1} \\
        &= \mB^T(\vx_p) \mM^{-1}(\vx_p).
    \end{aligned}
\end{equation}
}

\revision{
Applying \Eqref{eq:mls:mpm:linear:kernel:gradient} to the stacked function sample values at nodes $\vu$, we obtain a compact formula for both the reconstructed function value $\hat{u}_p$ and the gradient $\nabla_{\vz} \hat{u}_p$ of $u$:
\begin{equation}
    \label{eq:mls:linear:recon}
    \begin{aligned}
        \begin{bmatrix}
            \hat{u}_p        \\
            \nabla_{\vz} \hat{u}_p \\
        \end{bmatrix} &=
        [\vw_p, \mG_p]^T \vu = 
        \mM^{-1}(\vx_p) \mB(\vx_p) \vu.
    \end{aligned}
\end{equation}
}

We adopt the linear basis throughout this work.

\begin{figure}[tb]
    \centering
    \includegraphics[width=0.9\textwidth]{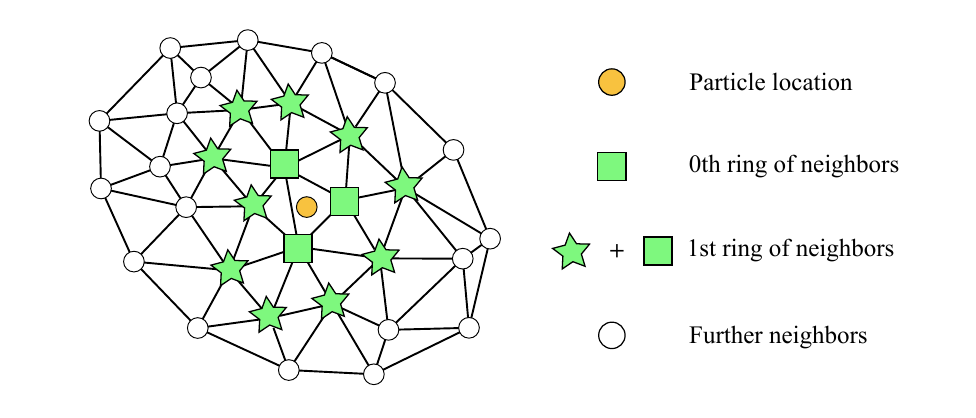}
    \caption{Schematic plot of the zeroth and first ring of neighbors.}
    \label{fig:neighbors}
\end{figure}

\subsubsection{Extending MLS-MPM Onto Unstructured Meshes}
\label{sec:interpws:mls:unstructured}

We select MLS-MPM as our foundation because of its inherent versatility, allowing it to be applied to adjacent nodes without reliance on specific topological or positional constraints. Our implementation and experiments are based on triangular and tetrahedral cells. Nonetheless, it is worth noting that our method can easily be extended to any tessellation by designing a smooth and locally diminishing function $\eta_v$ compatible with the tessellation, such as the one in \Eqref{eq:invD:diminish} for simplex cells.

\paragraph{Identifying Nearby Nodes Around a Particle}
To determine the nearby nodes for a given particle $p$, we first locate the cell that encompasses $p$ and refer to its nodes as $\gN_p^0$, representing the 0-ring neighbors of $p$. Then, we define $\gN_p^1$ as the 1-ring neighbors, which comprise all nodes connected to $\gN_p^0$. Note that $\gN_p^0 \subset \gN_p^1$. Similarly, we can define $\gN_p^2, \dots$ in an analogous manner, as illustrated in Figure~\ref{fig:neighbors}.

\revision{To quickly search for the cell that contains $p$, we pre-store the adjacency relationship between the spatial hashing grid and the mesh cells. When given a $\vx_p$, the spatial hash grid is queried, and we then only check the cells adjacent to this hash grid. The detailed pipeline can be found in \Appref{appdx:pipeline:hash:search}.}

\begin{figure}[tb]
    \centering
    \includegraphics[width=0.9\textwidth]{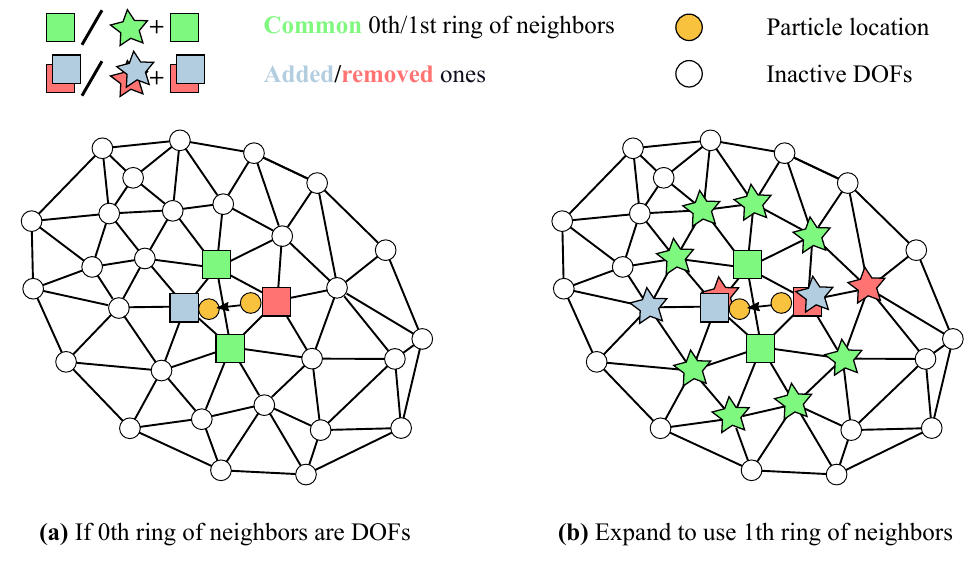}
    \caption{(a) When $\gN_p^0$ alone is selected as the active nearby nodes, as a particle crosses the cell edge, the nodes indicated by the blue and red boxes are added or removed, respectively. Consequently, the weights there must approach zero to ensure $\gC^0$ continuity, resulting in kernel degeneration along the edge. (b) Advancing to $\gC^1$ addresses this issue by incorporating a sufficient number of surrounding nodes to fully encompass the particle.}
    \label{fig:choose:ring}
\end{figure}

\paragraph{Ring Level Selection for Nearby Nodes}
When a specific level of ring neighbors is chosen as the active set of nodes, a natural question arises:

\textit{What is the minimum number of rings required to satisfy the desired properties of the MPM kernel?}

Assume $\gN_p^0$ is selected, and the particle only affects the nodes $i \in \gN_p^0$; since at least $\gC^0$ continuity is required for the kernel, when the particle passes one interface of the cell, the node not on the interface is removed from the active set and the kernel for it must be zero. This leads to the kernel degenerating, i.e., the kernel affecting merely the interface when the particle crosses it, as depicted in Figure~\ref{fig:choose:ring}a. Conversely, opting for 1-ring neighbors, $\gN_p^1$, effectively circumvents this issue, ensuring a non-degenerate kernel interaction as illustrated in Figure~\ref{fig:choose:ring}b.

\paragraph{Computing the Weights}
For conciseness, we replace the function input with the subscript, for example replacing $(\vx_p)$ with $p$ in all related equations, the reconstruction reads:
\begin{equation}
    \label{eq:mls:linear:ours:recon}
    \begin{aligned}
        \begin{bmatrix}
            \hat{u}_p        \\
            \nabla_{\vz} \hat{u}_p \\
        \end{bmatrix} =
        \mM^{-1}_p \mB_p \vu.
    \end{aligned}
\end{equation}

\subsubsection{Required Properties for the Transfer Kernel}
Consider the essential desirable properties for an MPM kernel:
\begin{enumerate}
\item The kernel must be a non-negative partition of unity. This means that the sum of the kernel weights for all nearby vertices of a particle should equal 1; i.e., $\sum_{v \in \gN_p^1} w_v = 1$, with each individual weight $w_v \geq 0, \forall v \in \gN_p^1$.
\item There should be a continuous reconstruction of both the function value and gradient as the particle crosses the cell boundary.
\end{enumerate}
With MLS-MPM, the partition of unity is inherently assured by the characteristics of MLS~\cite{levin1998approximation}, and non-negativity is assured by the uniform sampling of grid nodes (i.e., no degenerate samples). Lastly, \revision{with uniform grid nodes}, MLS-MPM ensures continuous reconstruction by utilizing a B-spline for sample weighting and provides $\gC^1$ continuity.

\revision{
However, this property holds only under uniform grid nodes with spacing properly aligned with the support of the B-spline weighting function. The key is that the B-spline function approaches zero for grid nodes that are about to be added or removed. Consequently, the influence of the discrete change in the set of active nodes on the assembly of $\mM_p$ and $\mB_p$ in \Eqref{eq:mls:linear:recon} is infinitesimal, ensuring no abrupt change during the reconstruction.
}

\revision{
Due to the varying spacing of the unstructured meshes, the weighting function is not guaranteed to approach zero for the added or removed nodes during cell crossings, leading to discontinuous reconstruction. This issue will be addressed in the next section.
}

\revision{
Still, we can borrow the key insight from MLS-MPM that “the weighting function for added or removed nodes should approach zero” and design a scheme to enforce this property. The solution will be discussed and presented in the next section.
}

\subsubsection{Remedying Discontinuous Reconstruction Across the Cell Boundary} 
\label{sec:interpws:mls:unstructured:remedy}
The jump change originates from the discrete change of the active set $\gN_p^1$ during particle cell crossing if their influence on the MLS assembly is nonzero. Hence, an intuitive solution is to artificially diminish their influence on the MLS assembly. To achieve this, we multiply any initial sample weighting function $d_{p, i}$, such as B-spline, by a smooth diminishing function $\eta_{p, i}$; i.e., $d'_{p, i} \leftarrow \eta_{p, i} d_{p, i}$. Here, $\eta_{p, i} \to 0$ for nodes that are added or removed from the active set $\gN_p^1$ during the cell crossing. A detailed proof of the efficacy of this approach is provided in \Appref{appdx:conservation:proof}.

For simplex elements, we design the following $\eta_{p, i}$:
\begin{equation}
    \label{eq:invD:diminish}
    \eta_{p, i} = \sum_{n \in \gN_p^0} B_{p, n} \mA_{i,n},
\end{equation}
\revision{where $\mA$ denotes the mesh's adjacency matrix; an adjacency matrix is a binary matrix representing the connectivity of a graph, where each $\mA_{i,j}=1$ indicates the presence of an edge between nodes $i$ and $j$. A mesh can naturally be viewed as a graph by connecting an edge between every pair of adjacent nodes in a cell. $B_{p, n}$ represents the barycentric coordinate for particle $p$ with respect to a specific node $n \in \gN_p^0$. Taking a 2D simplex, the triangular cell, as an example, the barycentric coordinates of a location $\vx_p$ are triplets of numbers $b_1, b_2, b_3$, subject to $b_1 + b_2 + b_3 = 1$. If these values are considered as masses placed at the nodes of the triangle, the centroid of these masses will be at $\vx_p$. Generally, barycentric coordinates can be calculated as follows:
}
\revision{
\begin{equation}
    \label{eq:barycentric}
    B_{p, n_i} = \frac{\det\left([\vx_{n_1}, \ldots, \vx_{n_{i-1}}, \vx_p, \vx_{n_{i+1}}, \ldots, \vx_{n_{\text{dim}+1}}]\right)}{\det\left([\vx_{n_1}, \ldots, \vx_{n_{\text{dim}+1}}]\right)}.
\end{equation}
}

Combining \eqref{eq:barycentric} with the adjacency matrix definition provides a more geometric interpretation of the design \eqref{eq:invD:diminish}: for $i \in \gN_p^1$, $\eta_{p, i}$ is the sum of the barycentric weights for all $n \in \gN_p^0$ that are connected to $i$. Note that $\eta_{p, i}=1,~\forall v \in \gN_p^0$. \Appref{appdx:conservation:proof} proves the claimed diminishing property for this design in the simplex cell, while \Figref{fig:invD:diminish} provides a simple visual illustration of the proposed \(\eta_{p, i}\).

\subsubsection{Verification of the Proposed Kernel}
\label{sec:verify:continous:kernel}
To verify that the proposed method can produce continuous reconstruction, analytical and numerical solutions of some examples are produced in 1D and 2D test cases, respectively.

\begin{figure}[tb]
    \centering
    \includegraphics[width=0.9\textwidth]{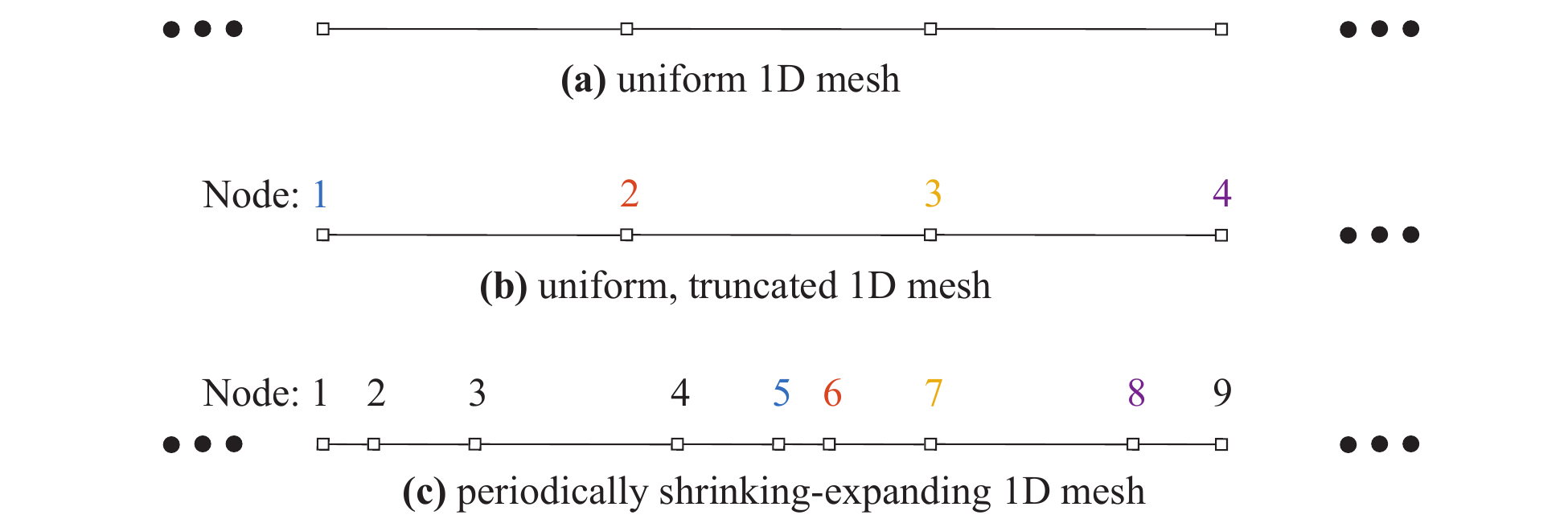}
    \caption{1D meshes: (a) Uniform. (b) Uniform but truncated. (c) Periodically shrinking/expanding.}
    \label{fig:1d_w_setup}
\end{figure}

For the 1D case, the first basic verification is conducted on a uniform mesh, as shown in Figure~\ref{fig:1d_w_setup}a. Figure~\ref{fig:1d_w_uniform_compare}a shows the correct kernel reconstruction with the diminishing function $\eta$, while Figure~\ref{fig:1d_w_uniform_compare}b, as an ablation, shows that the reconstruction is discontinuous even for the simplest uniform mesh, proving the necessity of $\eta$. The detailed setup for this analytical solution is provided in \Appref{appdx:kernel:verify:misc}.

\begin{figure}[tb]
    \centering
    \includegraphics[width=0.9\textwidth]{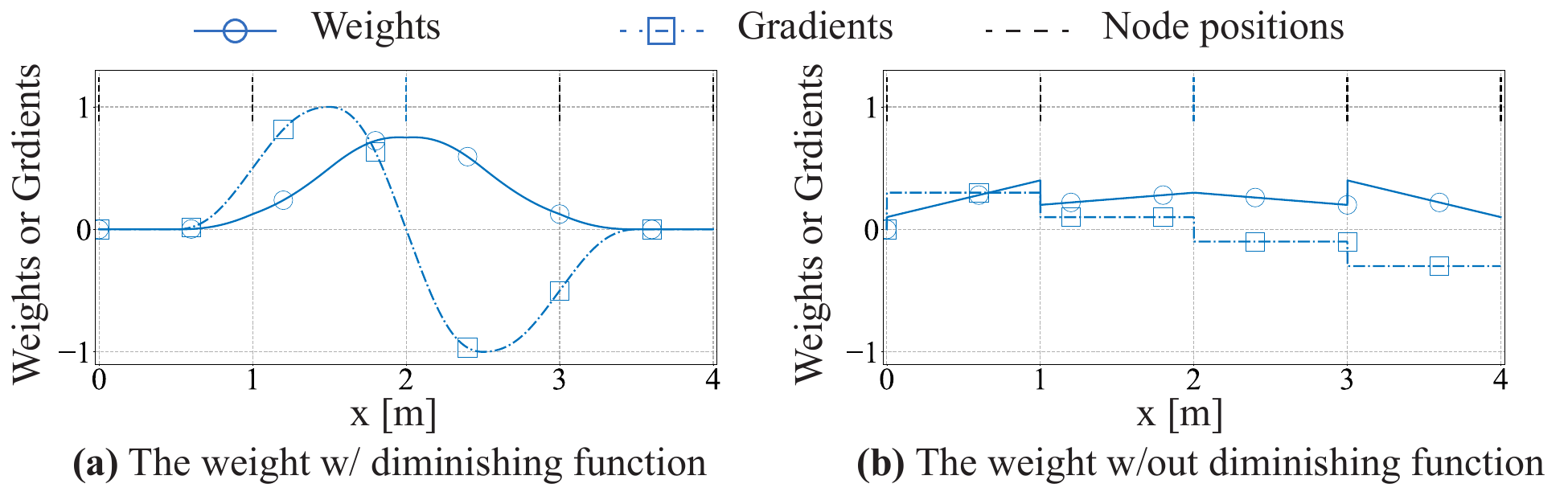}
    \caption{Comparison of kernel values and gradient estimations on a uniform 1D mesh (a) with and (b) without applying the diminishing function.}
    \label{fig:1d_w_uniform_compare}
\end{figure}

\begin{figure}[tb]
    \centering
    \includegraphics[width=0.9\textwidth]{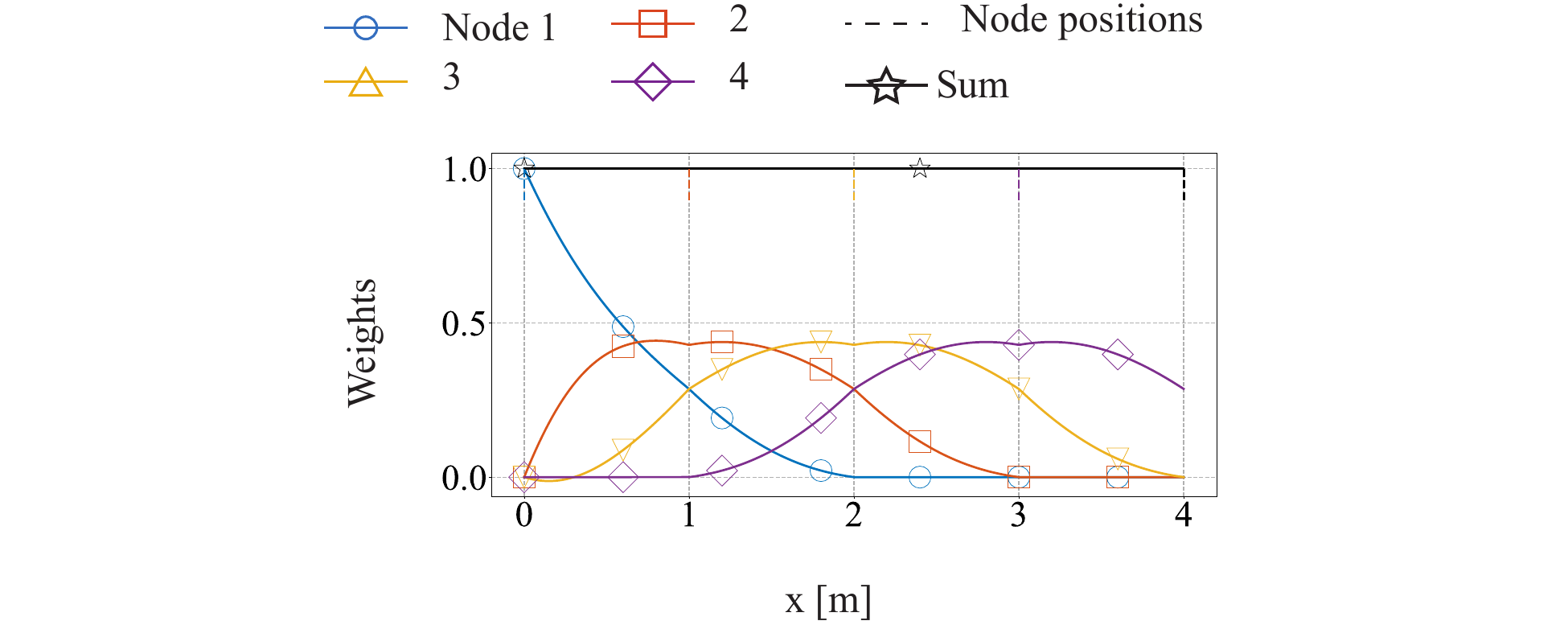}
    \caption{The negative weight for Node 3 (yellow) when the particle is in the boundary cell and there is no extra layer.}
    \label{fig:1d_w_uniform_truncated}
\end{figure}

Note that when a particle is in a boundary cell, such as Node 3 in Figure~\ref{fig:1d_w_uniform_truncated}, negative weight values may be obtained for some interior nodes. This is caused by kernel degeneration due to the absence of a first ring of neighbors on the boundary side during MLS sampling. To remedy this problem, which can cause numerical instabilities~\cite{andersen2010analysis}, an extra layer of cells beyond the real boundary is included in our experiments.

\begin{figure}[!t]
    \centering
    \includegraphics[width=0.9\textwidth]{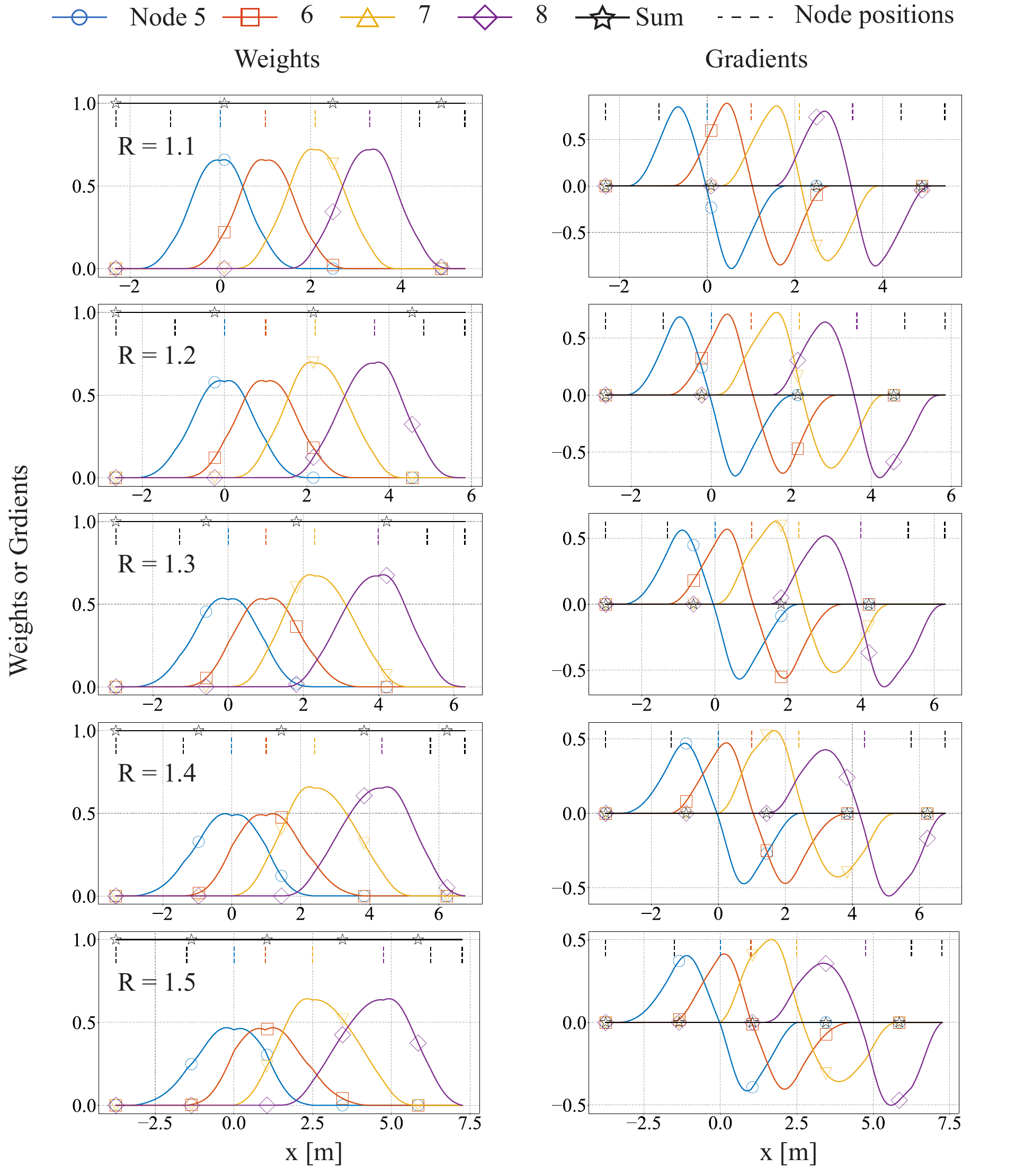}
    \caption{Kernel values (left column) and gradient estimations (right column) on a periodically shrinking/expanding 1D mesh with varying size transition rate $R$.}
    \label{fig:1d_w_periodic}
\end{figure}

The next verification is on a periodically shrinking and expanding 1D mesh (Figure~\ref{fig:1d_w_setup}c). The mesh contains cyclic cell sizes of $[\dots,1,R,R^2,R,1,\dots]$ designed to mimic the transition between varying mesh resolutions. The size transition ratios tested range from $1.1$ to $1.5$ to correspond with typical transition ratios in FEM analysis. Kernel reconstructions are conducted on Nodes 5, 6, 7, and 8 as a full cycle. As shown in Figure~\ref{fig:1d_w_periodic}, both the kernel and the gradient estimations are piece-wise $\gC^1$.

\revision{
We note that kinks can be seen in the function value plots in \Figref{fig:1d_w_periodic}, leading to the potential confusion that the gradient is discontinuous. However, the plots of kernel values are $\vw(\vz ; \vx_p)|_{\vz=\vx_p} = \vw(\vx_p ; \vx_p)$ versus $\vx_p$, as in \Eqref{eq:kernel:define}; for the misleading statement to hold true, the gradient should have been taken with respect to the plot axis $\vx_p$, i.e., $\nabla_{\vx_p} \vw(\vx_p ; \vx_p)$, which is not the case according to \Eqref{eq:kernel:gradient:define}.
}

The ablation tests are also performed on a 2D unstructured mesh featuring a ``\&'' shape. The comparison between scenarios with and without the use of \(\eta\), as shown in Figure~\ref{fig:2d_w_compare}a and Figure~\ref{fig:2d_w_compare}b respectively, validates the importance of \(\eta\) and the effectiveness of the proposed method in managing unstructured meshes.

\revision{
Finally, we experimentally show that \ours \ can seamlessly be combined with other schemes, such as the Affine Particle in Cell (APIC) scheme \cite{jiang2015affine,jiang2017angular} to help conserve the total angular momentum of the system. For details, see Appendix \ref{appdx:apic:conservation}.
}

\begin{figure}[tb]
    \centering
    \includegraphics[width=0.9\textwidth]{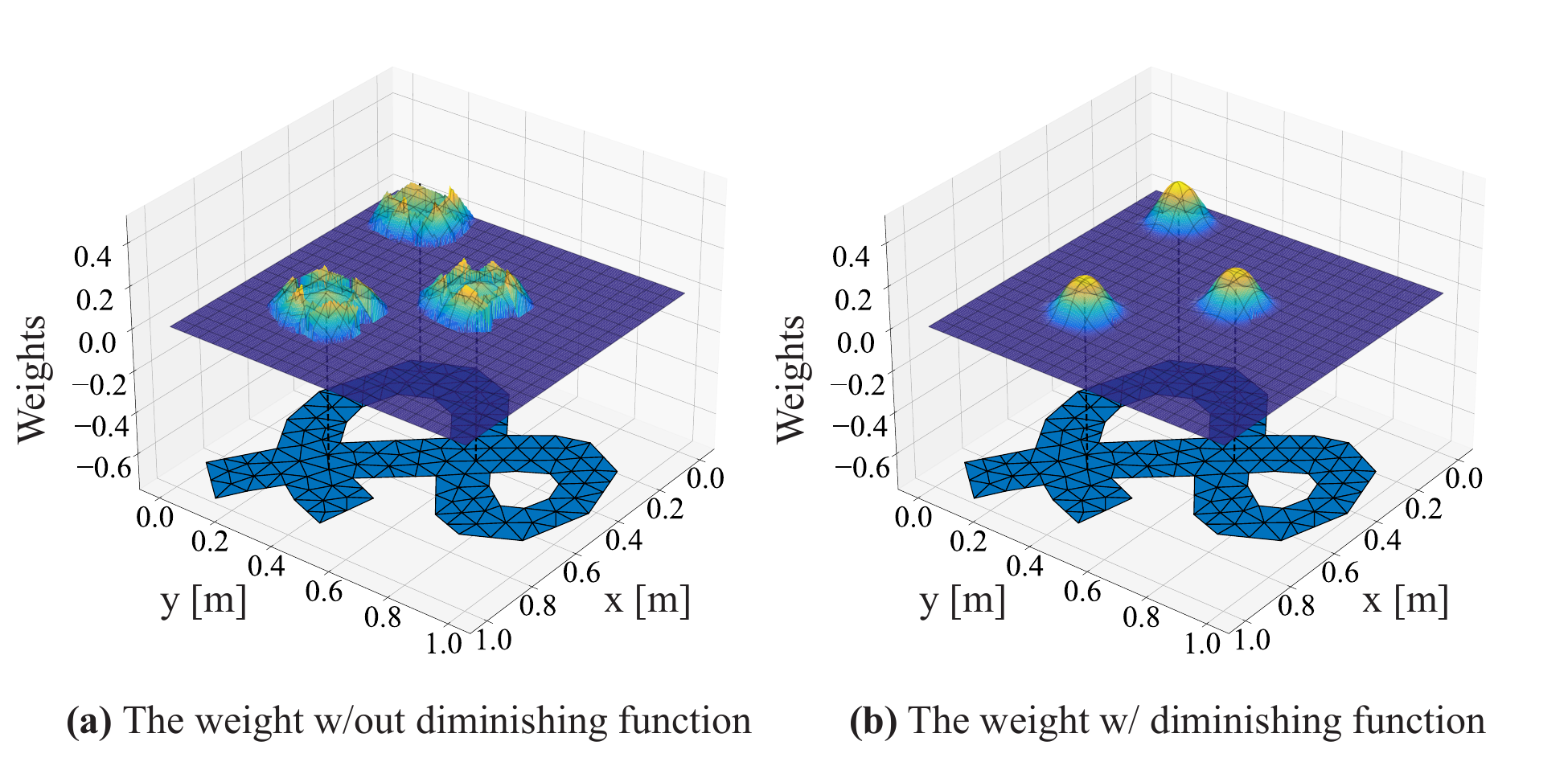}
    \caption{Comparison of the kernel on an unstructured mesh (a) without and (b) with the application of the diminishing function.}
    \label{fig:2d_w_compare}
\end{figure}

\section{Experiments and Results}
\label{sec:experiments}

To demonstrate and assess the effectiveness of our approach, particularly its reduced cross-cell error owing to the continuous gradient reconstruction, we have chosen representative test cases from prior related studies. \revision{Our benchmarking relies on analytical solutions when feasible; alternatively, we use the standard MPM with B-spline or GIMP basis functions at a sufficiently high resolution.} All experiments were carried out on a single PC equipped with an Intel\textsuperscript{\textregistered} Core\texttrademark\ i9-10920X CPU.

\subsection{1D Vibrating Bar}

\begin{figure}[tb]
    \centering
    \includegraphics[width=0.9\textwidth]{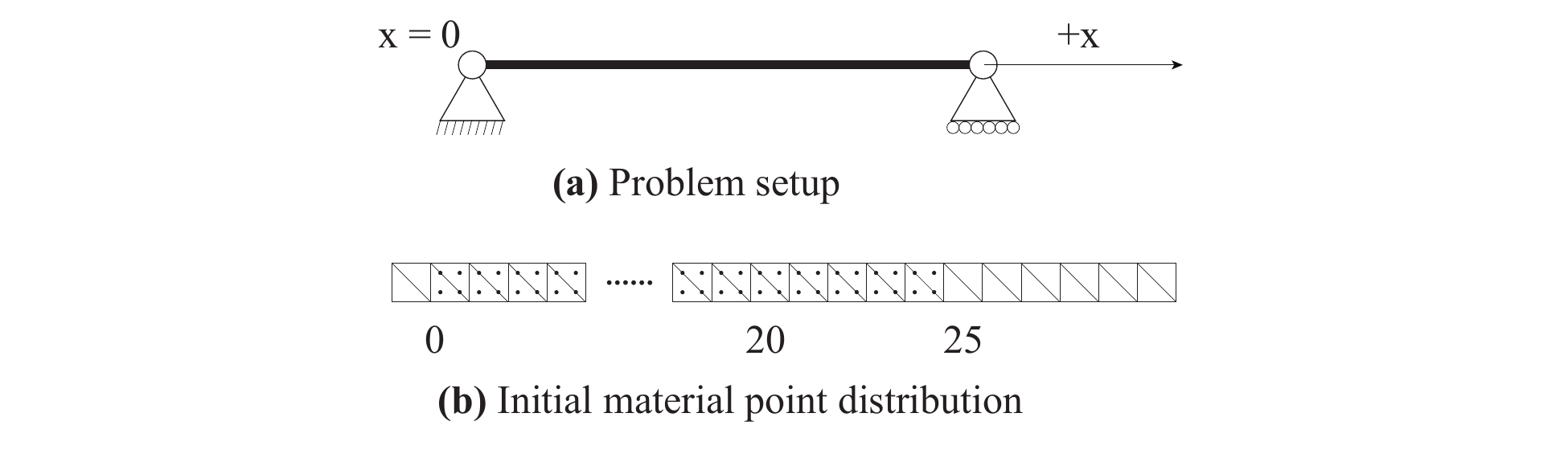}
    \caption{Setup of the 1D bar vibration test.}
    \label{fig:1d_bar}
\end{figure}

Consider the 1D vibration bar problem shown in Figure~\ref{fig:1d_bar}a \cite{wilson2021distillation}. The left end of the bar is fixed and the right has a sliding condition in the $x$ direction. The physical properties of the bar are: $E=100$\,Pa, $\nu = 0$, $L=25$\,m, and $\rho=1$\,kg/m$^3$. The initial velocity conditions are $\dot{u}(x, t=0)=v_0 \sin \left(\beta_1 x\right)$ with $\beta_1=\frac{\pi}{2 L}$.

\begin{figure}[tb]
    \centering
    \includegraphics[width=0.9\textwidth]{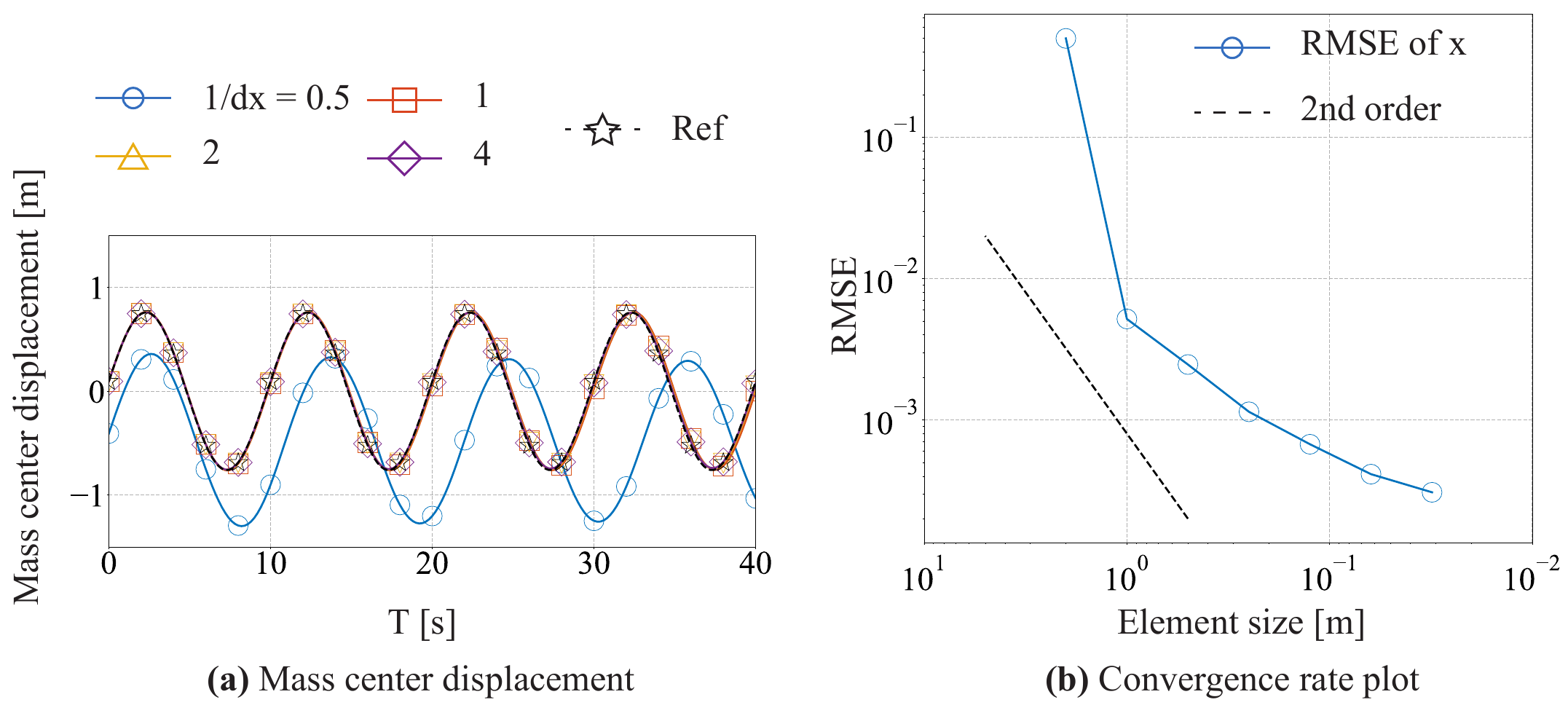}
    \caption{Plots of (a) the center of mass displacement of the bar and (b) convergence rate of the RMSE of particle displacements.}
    \label{fig:1d_bar_convergence}
\end{figure}

The analytical expression of the center of mass in this problem is
\begin{equation}
    x(t)_{C M}=\frac{L}{2}+\frac{v_0}{\beta_1 L \omega_1} \sin \left(\omega_1 t\right),
\end{equation}
and
\begin{equation}
    \dot{u}(t)_{C M}=\frac{v_0}{\beta_1 L} \cos \left(\omega_1 t\right),
\end{equation}
with $\omega_1=\beta_1 \sqrt{E / \rho}$.

The original experiments in~\cite{wilson2021distillation} included two velocity settings: $v_0=0.1$\,m/s and $v_0=0.75$\,m/s. The lower velocity setting, $v_0=0.1$\,m/s, was utilized solely for validation against the linear kernel MPM, as it does not involve cell crossings. Here, we focus on the higher-velocity setting to assess the effectiveness of \ours\ in addressing cell-crossing errors.

Figure~\ref{fig:1d_bar_convergence} presents the convergence rate of \ours\ with grid refinement. Specifically, Figure~\ref{fig:1d_bar_convergence}a shows that, with the exception of the coarsest resolution $dx=2$\,m, \ours\ consistently achieves high accuracy, with a maximum root mean square error (RMSE) of $0.554\%$ in particle displacements. Figure~\ref{fig:1d_bar_convergence}b indicates that the convergence rate is approximately second order on coarser grids, but it starts to level off on finer grids due to mounting temporal errors, aligning with established MPM theory~\cite{jiang2016material}.

\begin{figure}[tb]
    \centering
    \includegraphics[width=0.9\textwidth]{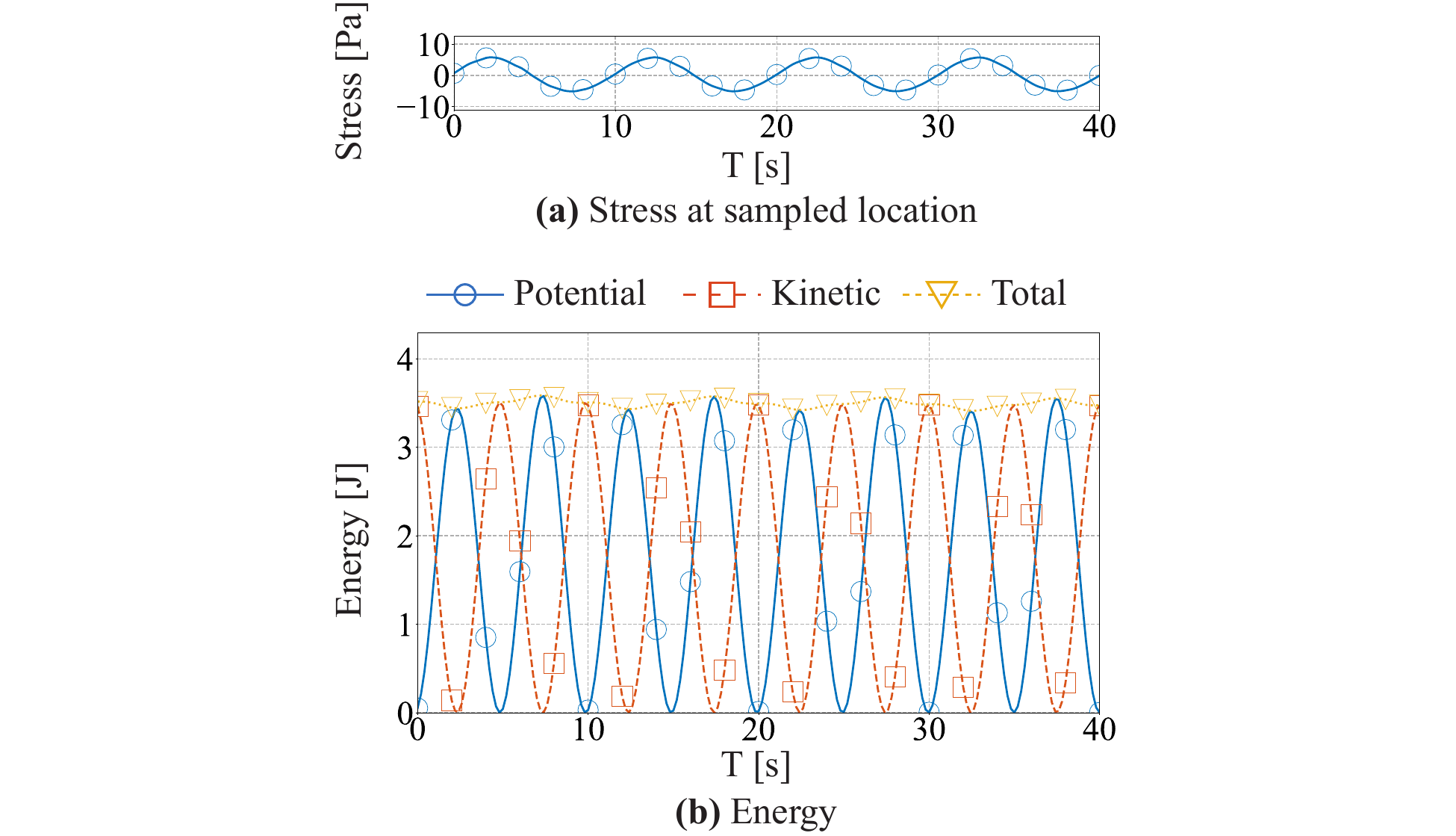}
    \caption{Plots of (a) the stress at the sampled particle closest to $[17.5, 0.5]$ and (b) the system energy.}
    \label{fig:1d_bar_stress_eng}
\end{figure}

Figure~\ref{fig:1d_bar_stress_eng}a displays the stress profile for a particle located at $x_0 = 12.75$\,m, which undergoes the most frequent cell crossings during its vibrational motion. The outcomes achieved with \ours\ showcase a remarkable level of smoothness and precision. Figure~\ref{fig:1d_bar_stress_eng}b illustrates the energy dynamics for the entire system, revealing that the system's energy is largely conserved throughout the simulation, with only slight fluctuations.
\revision{
We believe the fluctuations in the energy plot are due to the symplectic integration schemes or the combined effect of the FLIP scheme. Similar phenomena have been observed in previous works~\cite{donnelly2005symplectic} and~\cite{tran2019temporal}, respectively.
}
These findings collectively underscore the robustness and precision of \ours\ in managing intense cell crossings by particles.

\subsection{2D Collision Disks}

\begin{figure}[tb]
    \centering
    \includegraphics[width=0.9\textwidth]{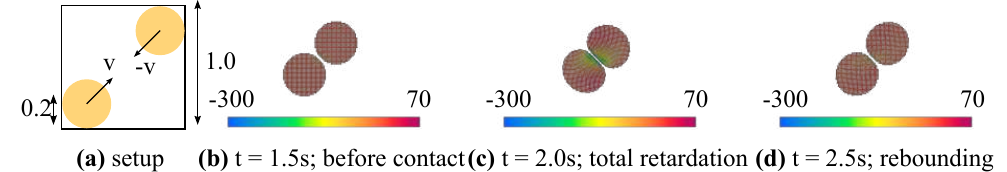}
    \caption{2D collision disks: (a) Problem setup. (b)--(d) Snapshots of the simulation at 1.5s, 2.0s, and 2.5s.}
    \label{fig:2d_disk}
\end{figure}

Next, we considered the problem of two colliding elastic disks shown in Figure~\ref{fig:2d_disk}a \cite{wilson2021distillation}. The physical properties of the disks are: $E=1000$\,Pa, $\nu = 0.3$, $\rho=1000$\,kg/m$^3$, and $\vv= \pm (0.1,0.1)$\,m/s for the left and right disks, respectively. Each disk was discretized with $462$ material points using the triangle mesh of a disk. The background mesh was generated using Delaunay triangulation with a target element size of $0.025$\,m. We plot key snapshots of the simulation in Figure~\ref{fig:2d_disk}b--d, with the impact at $1.5$\,s, total retardation right before $2.0$ s, and rebounding separation right before $2.5$\,s.

\begin{figure}[tb]
    \centering
    \includegraphics[width=0.9\textwidth]{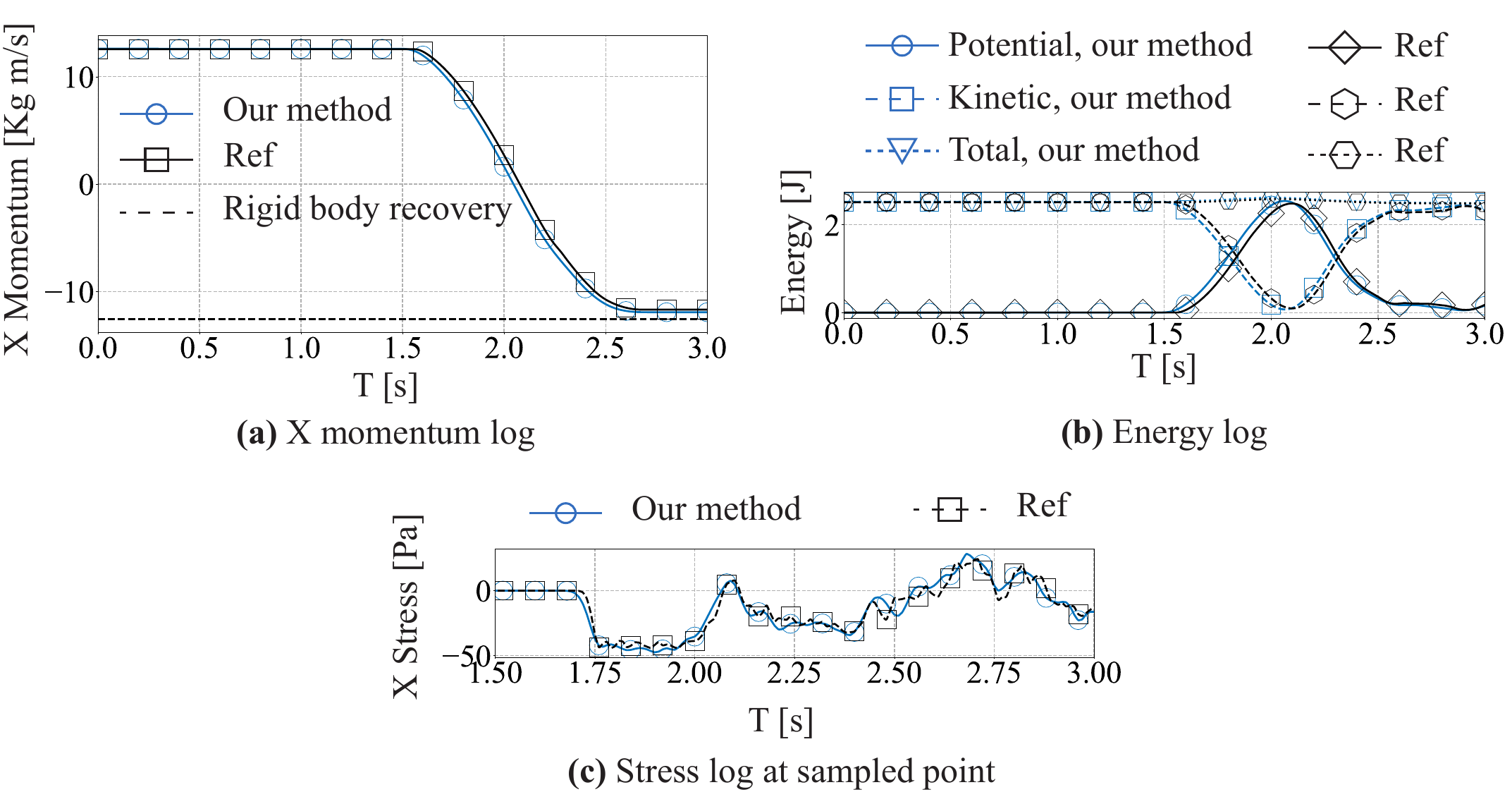}
    \caption{Plots of (a) the momentum in the $x$-direction of the left disk, (b) the energies of the system, and (c) the stress at the sampled particle closest to the center of the left disk.}
    \label{fig:2d_disk_plots}
\end{figure}

Quantitative results for the collision disks are presented in Figure~\ref{fig:2d_disk_plots}. In Figure~\ref{fig:2d_disk_plots}a, a comparison of momentum recovery during collision between \ours\ and the B-spline MPM with sufficiently high resolution is shown. While a perfect momentum recovery, such as that in the rigid collision (dashed gray line in Figure~\ref{fig:2d_disk_plots}a), is not expected, \ours\ approaches this limit effectively.
Similarly, Figure~\ref{fig:2d_disk_plots}b displays the kinetic energy recovery during the collision. The results indicate that \ours\ effectively preserves the system energy.
Figure~\ref{fig:2d_disk_plots}c illustrates the stress log at the center particle of the left disk. The results align perfectly with the reference, but only for negligible fluctuations, showing that \ours\ does not generate spurious stress oscillations either from the collision or cell crossings.

\subsection{2D Cantilever With Rotations}

\begin{figure}[tb]
    \centering
    \includegraphics[width=0.9\textwidth]{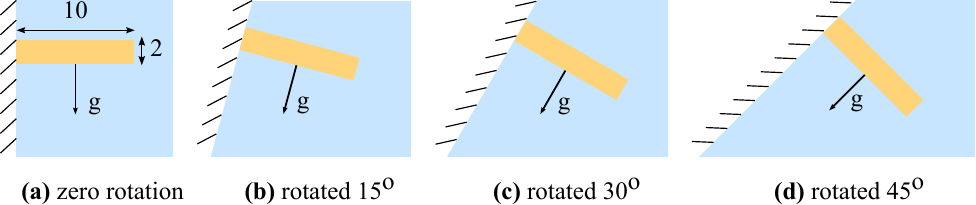}
    \caption{2D cantilever problem under different rotation angles: (a) $\SI{0}{\degree}$, (b) $\SI{15}{\degree}$, (c) $\SI{30}{\degree}$, and (d) $\SI{45}{\degree}$.}
    \label{fig:2d_cantilever}
\end{figure}

Although an unstructured mesh offers the adaptability to match any boundary shape, the cell orientation, or a different tessellation, can potentially affect accuracy. To illustrate the precision of our method under various rotation angles, we examined the case of a cantilever under its own weight, as shown in Figure~\ref{fig:2d_cantilever}a \cite{wilson2021distillation}. The cantilever's physical characteristics are as follows: length $l= 10$\,m, height $h = 2$\,m, gravitational acceleration $g=9.81$\,m/s$^2$, Young's modulus $E=100000$\,Pa, Poisson's ratio $\nu = 0.29$, and density $\rho=2$\,kg/m$^3$. The cantilever was discretized with uniformly spaced particles in both directions. We created the background mesh using Delaunay triangulation, aiming for an element size of $0.5$\,m. Additionally, we rotated the mesh of the cantilever by angles of $15^\circ$, $30^\circ$, and $45^\circ$ to showcase the resilience of our method to rotation, as depicted in Figure~\ref{fig:2d_cantilever}b--d.

\begin{figure}[tb]
    \centering
    \includegraphics[width=0.9\textwidth]{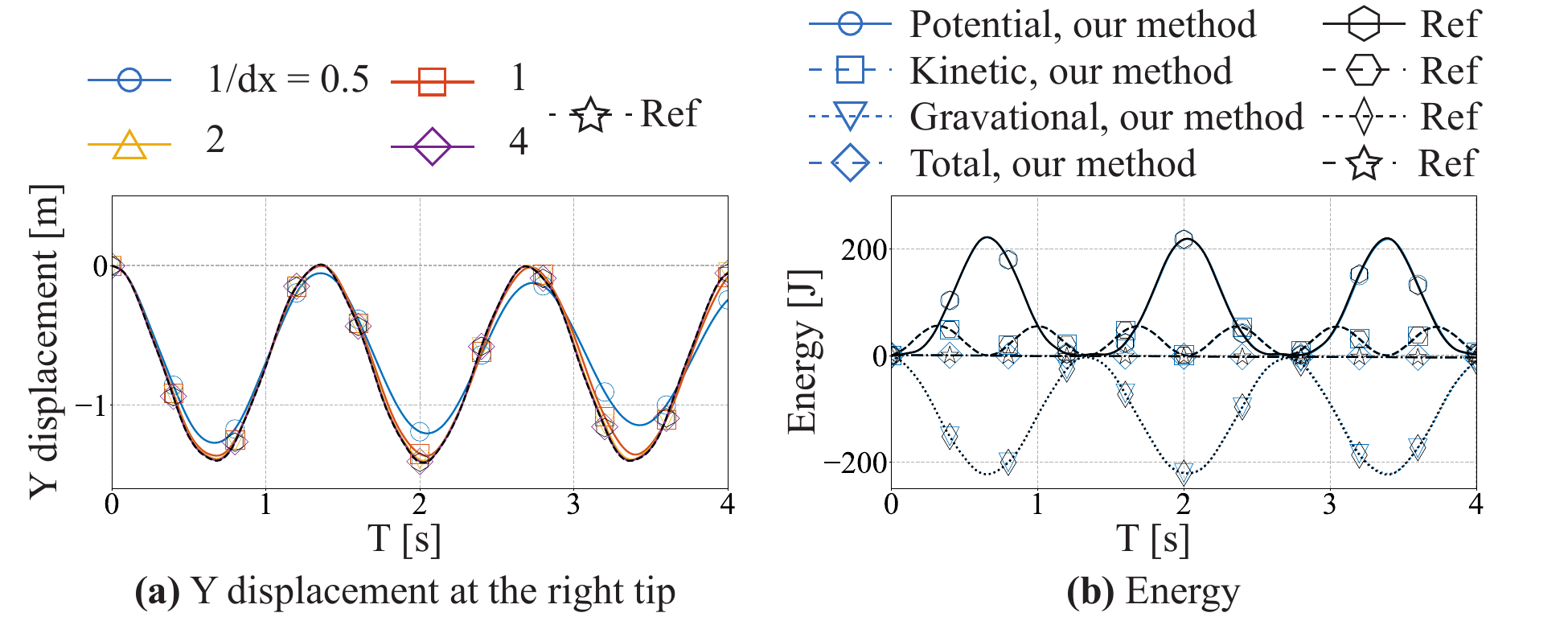}
    \caption{Plots of (a) the displacement in the $y$-direction at the right tip of the cantilever and (b) the energies of the system.}
    \label{fig:2d_cantilever_y_eng}
\end{figure}

Figure~\ref{fig:2d_cantilever_y_eng}a illustrates the spatial convergence of the $y$-displacement at the right tip of the cantilever beam under grid refinement. Notably, except for the coarse resolutions of $dx = 2$\,m and $dx = 1$\,m, errors for all finer resolutions are negligible. Therefore, a resolution of $dx = 0.5$\,m was employed to ensure sufficient accuracy for all subsequent plots in this experiment. Figure~\ref{fig:2d_cantilever_y_eng}b demonstrates that \ours\ effectively conserves energy, aligning with the reference B-spline MPM.

\begin{figure}[tb]
    \centering
    \includegraphics[width=0.9\textwidth]{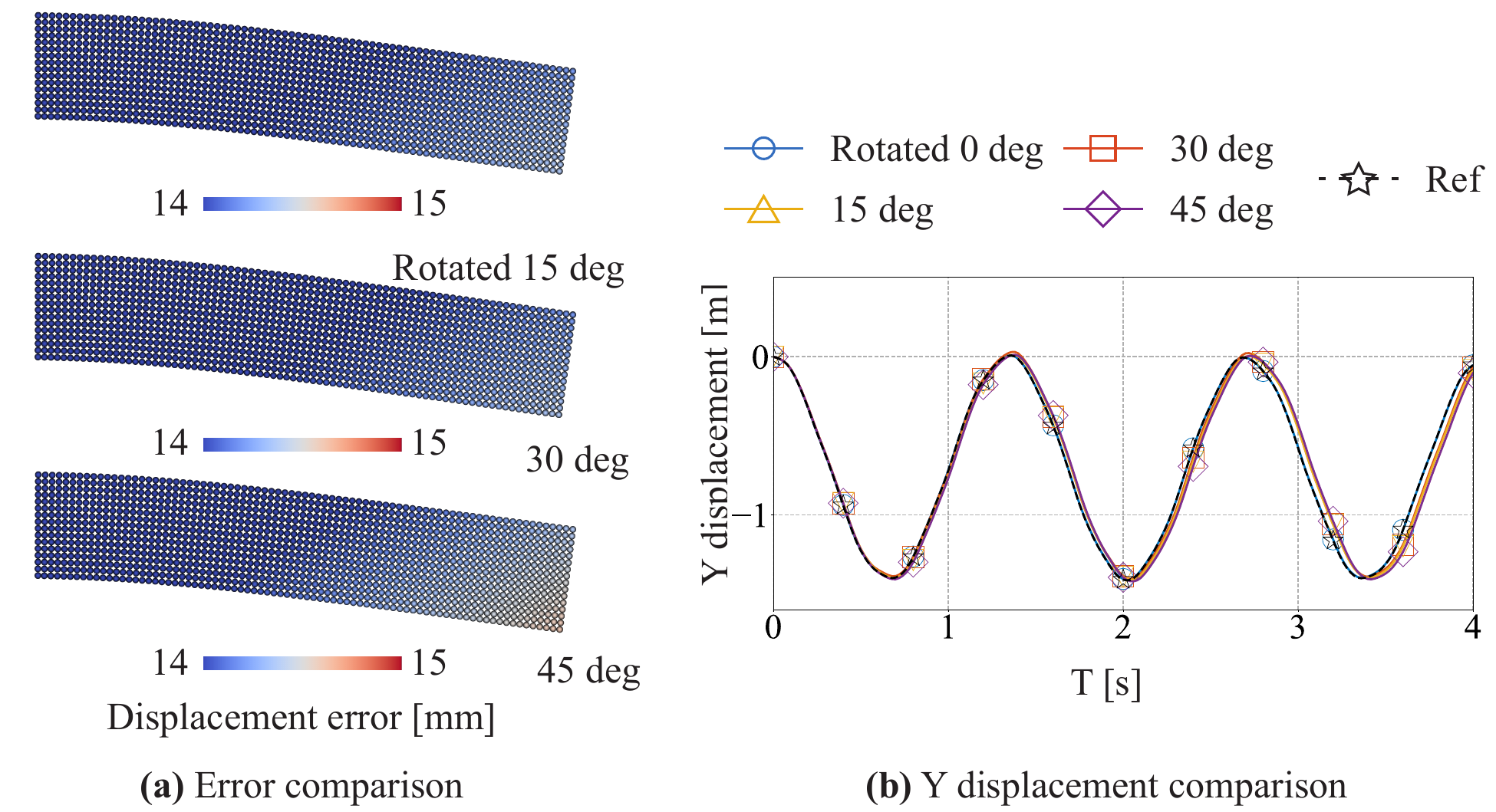}
    \caption{(a) 
    Snapshots of the cantilever with different initial rotating angles. (b) Comparison of the displacement in the $y$-direction at the right tip of the cantilever.}
    \label{fig:2d_cantilever_rotated}
\end{figure}

Figure~\ref{fig:2d_cantilever_rotated}a shows snapshots of the cantilever with different initial mesh rotation angles. The results indicate that \ours\ is robust under mesh rotation with only minor visible errors. Figure~\ref{fig:2d_cantilever_rotated}b quantitatively compares the $y$-displacement at the right tip. The results align well overall with both zero rotation and the reference, with errors of $1.27\%$, $2.18\%$, and $4.72\%$ for $15^\circ$, $30^\circ$, and $45^\circ$ rotation, respectively.

\begin{figure}[tb]
    \centering
    \includegraphics[width=0.9\textwidth]{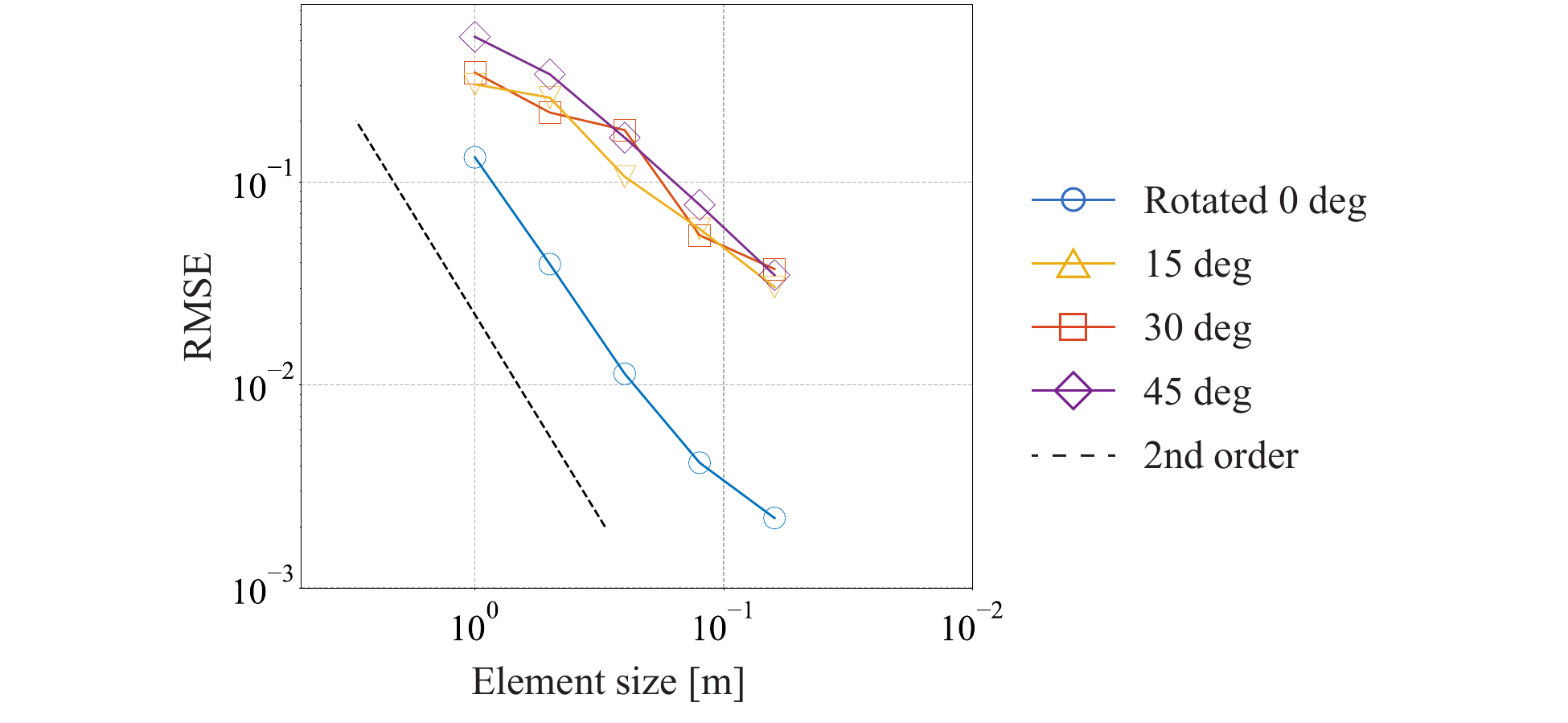}
    \caption{Convergence plot of the RMSE of particle displacements.}
    \label{fig:2d_cantilever_convergence}
\end{figure}

The convergence rate of \ours\ is demonstrated in Figure~\ref{fig:2d_cantilever_convergence}. The results indicate that for cases with zero rotation, the convergence rate is second order. While the RMSE increases slightly for cases with mesh rotation, it still remains in the magnitude of $1E-2$, and the convergence rate remains near second order.
These combined results demonstrate the robustness and accuracy of \ours\ under mesh rotation.

\revision{
\subsection{2D Ball in a Wavy Channel}
}

\revision{
\begin{figure}
    \centering
    \includegraphics[width=0.9\textwidth]{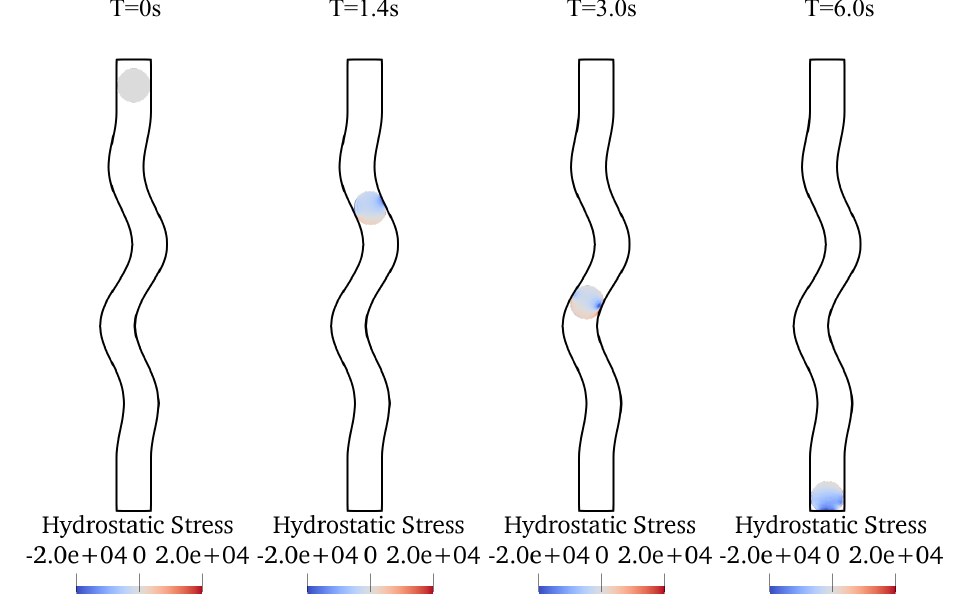}
    \caption{2D ball in a wavy channel: snapshots of the simulation at 0s, 1.4s, 3.0s, and 6.5s.}
    \label{fig:2d_ball_wavy}
\end{figure}
}

\revision{
We highlight the proposed method's ability to conform to irregular geometric boundaries. To this end, we consider the case of a ball freely falling but confined in a wavy channel, as shown in the leftmost subfigure of \Figref{fig:2d_ball_wavy}. The physical properties of the ball are: radius \(r=1.0\) m, Young's modulus \(E=100\) kPa, Poisson's ratio \(\nu = 0.29\), and density \(\rho=400\) kg/m\(^3\). The ball was discretized with 4735 randomly sampled material points. The background wavy channel has a sinusoidal shape. 
}

\revision{
The left wall of the channel has an analytical expression of:
\begin{equation}
    \begin{aligned}
        \begin{array}{cc}
            x_l= &
            \begin{cases}
                A \sin (\omega_1 y) \sin (\omega_2 y), & 0 < y \leq 20.0 \text{ m} \\
                0, & \text{Otherwise}, \\
            \end{cases}
        \end{array}
    \end{aligned}
\end{equation}
where \(A=1.0\) m, \(\omega_1 = \frac{\pi}{5}\) rad/m, and \(\omega_2 = \frac{\pi}{20}\) rad/m. The right wall is created by shifting the left wall by 2.0 m, i.e., \(x_r = x_l + 2.0\) m. The background mesh was generated using Delaunay triangulation with a target element size of 0.05 m, resulting in 43360 cells.
}

\revision{
\begin{figure}[htbp]
    \centering
    \includegraphics[width=0.9\textwidth]{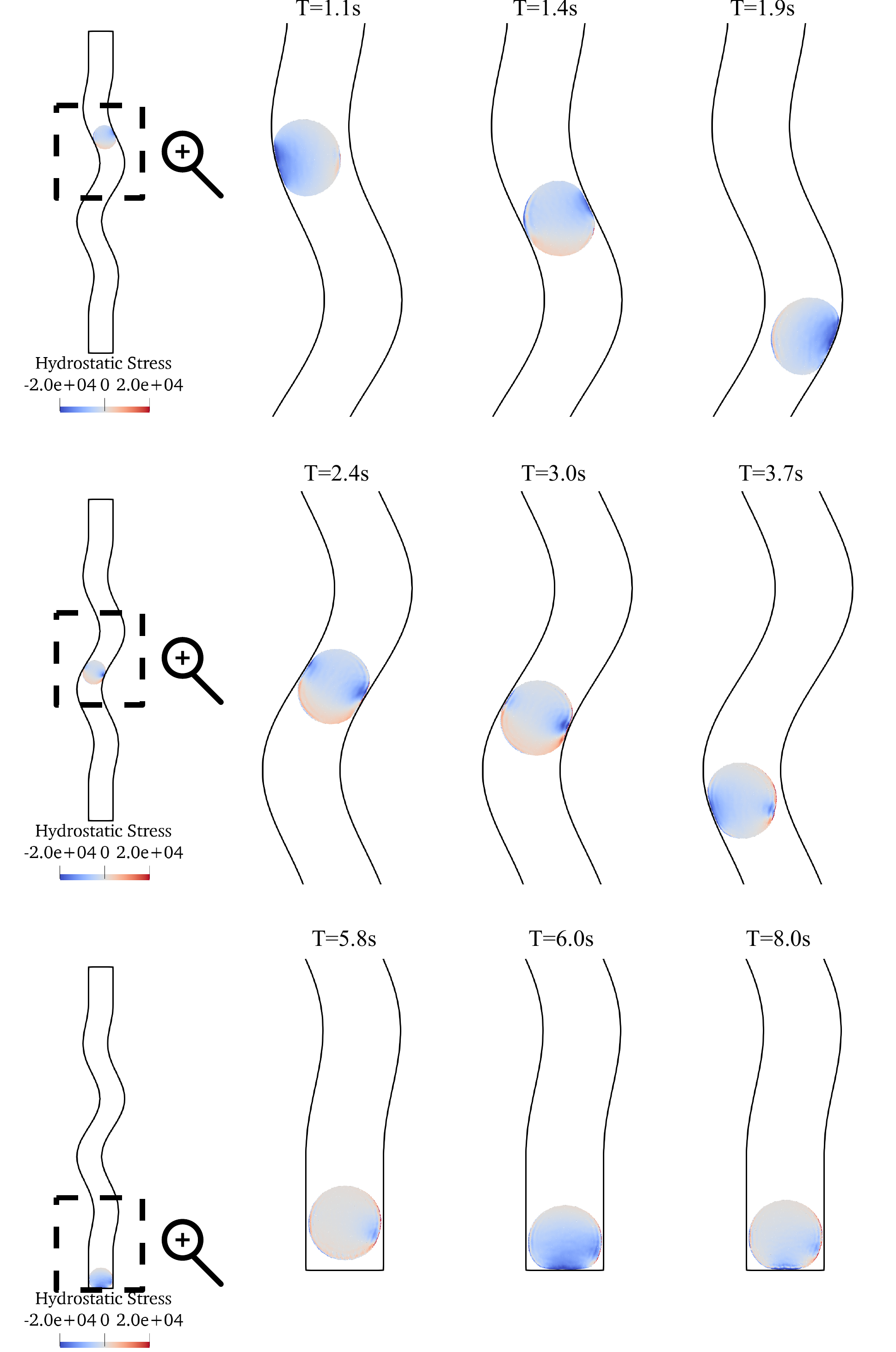}
    \caption{Zoomed-in view of the ball's deformation in the wavy channel at key timestamps. From top to bottom, timestamps around 1.4s, 3.0s, and 6.0s.}
    \label{fig:2d_ball_wavy_zoom}
\end{figure}
}

\revision{
The snapshots of the simulation at 1.4s, 3.0s, and 6.5s are shown in \Figref{fig:2d_ball_wavy}, while the zoomed-in views of the ball's deformation and hydrostatic stress are shown in \Figref{fig:2d_ball_wavy_zoom}. \ours \ captures both the bouncing into the wavy channel (at 1.1s, 1.9s, 3.7s) and the squeezing through the narrow part of the channel (at 1.4s, 2.4s, 3.0s) with no rasterization artifacts, proving the robustness of the proposed method in simulating under general mesh tessellation and handling irregular geometry boundaries.
}

\subsection{3D Slope Failure}
\label{sec:experiment:3d_slope}

Next, the performance of the proposed approach was investigated when dealing with material behavior involving plasticity.
To this end, we simulated failure of a 3D slope comprosed of sensitive clay. 
The problem geometry was adopted from \cite{zhao2023circumventing} and is illustrated in Figure~\ref{fig:3d_slope_geometry}.
Here, the bottom boundary of the slope is fixed and the three lateral sides are supported with rollers.
To model the elastoplastic behavior of the sensitive clay in an undrained condition, a combination of Hencky elasticity and J2 plasticity with softening was used.
The softening behavior is governed by the following exponential form: $\kappa = (\kappa_p - \kappa_r) e^{-\eta \varepsilon_q^p} + \kappa_r$, where $\kappa$, $\kappa_p$, and $\kappa_r$ denote the yield strength, the peak strength, and the residual strength, respectively, $\varepsilon_q^p$ denotes the equivalent plastic strain, and $\eta$ is a softening parameter.
The specific parameters were adopted from \cite{zhao2023circumventing}.
They are a Young's modulus of $E = 25$\,MPa, a Poisson's ratio of $\nu = 0.499$, a peak strength of $\kappa_p = 40.82$\,kPa, a residual strength of $\kappa_r = 2.45$\,kPa, and a softening parameter of $\eta = 5$.
The assigned soil density is $\rho = 2.15$\,t/m$^3$.

\begin{figure}[htbp]
    \centering
    \includegraphics[width=0.6\textwidth]{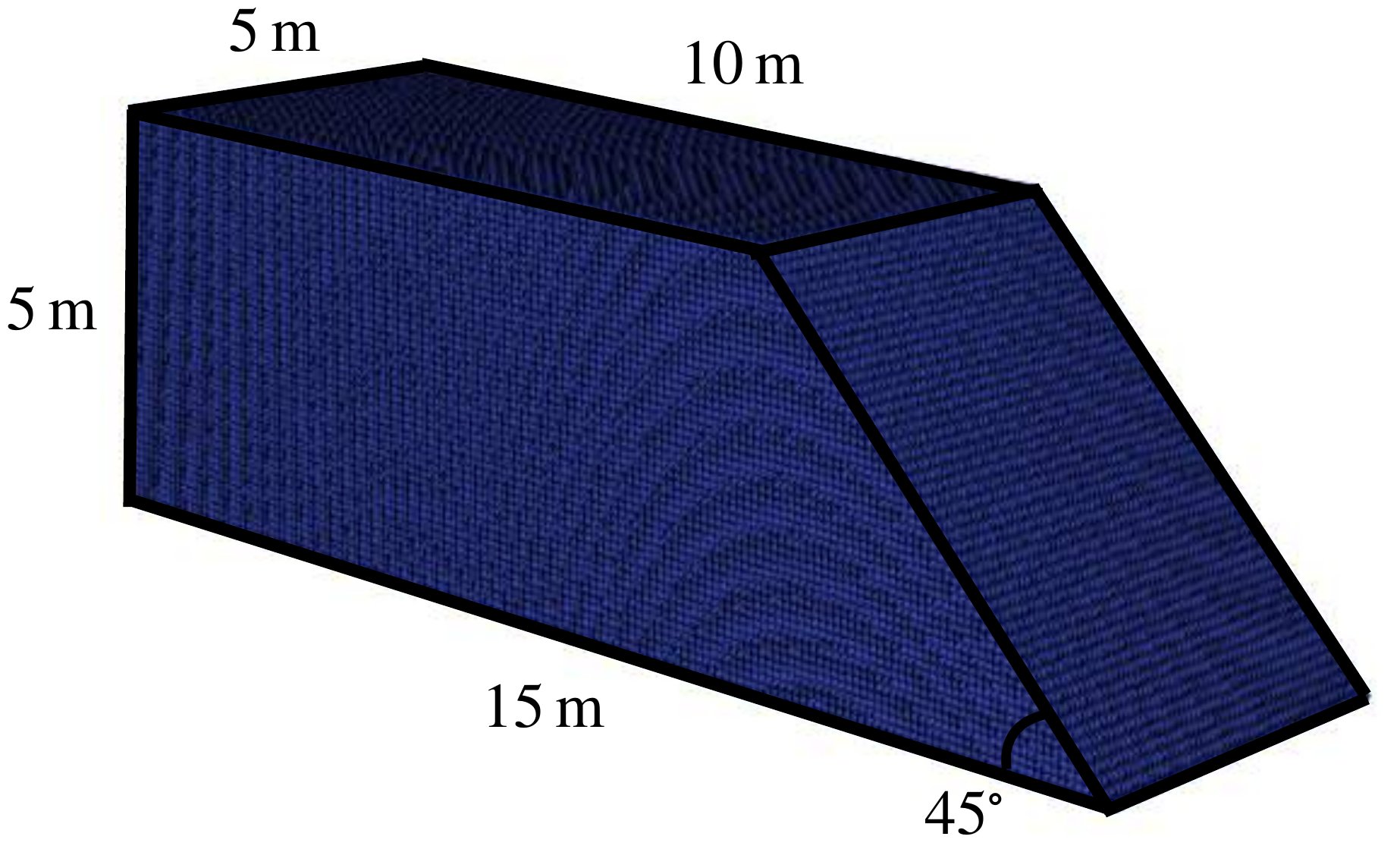}
    \caption{Problem geometry of the 3D slope failure (adapted from \cite{zhao2023circumventing}).}
    \label{fig:3d_slope_geometry}
\end{figure}

\begin{figure}[htbp]
    \centering
    \subfloat[$t=1.5$ s]{\includegraphics[width=0.9\textwidth]{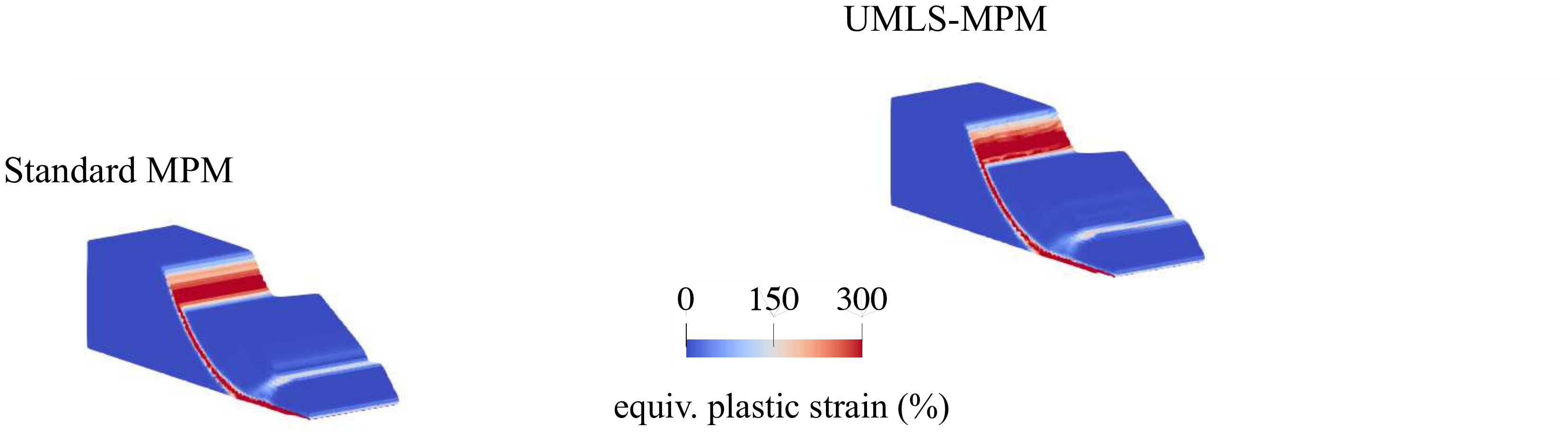}} \\
    \subfloat[$t=2.5$ s]{\includegraphics[width=0.9\textwidth]{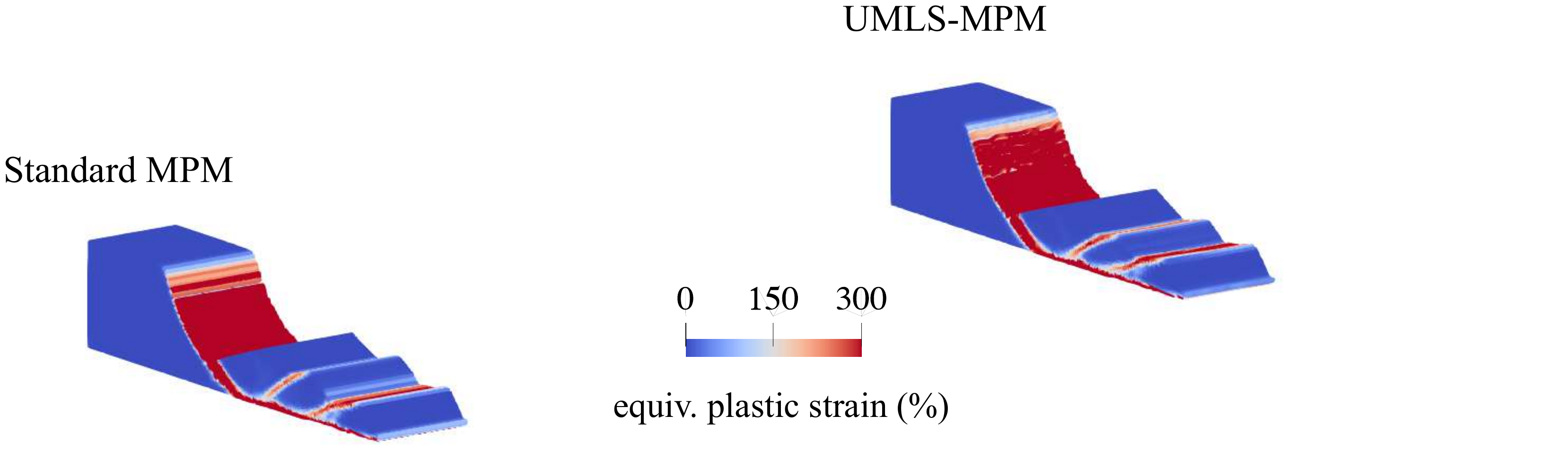}} \\
    \subfloat[$t=3.5$ s]{\includegraphics[width=0.9\textwidth]{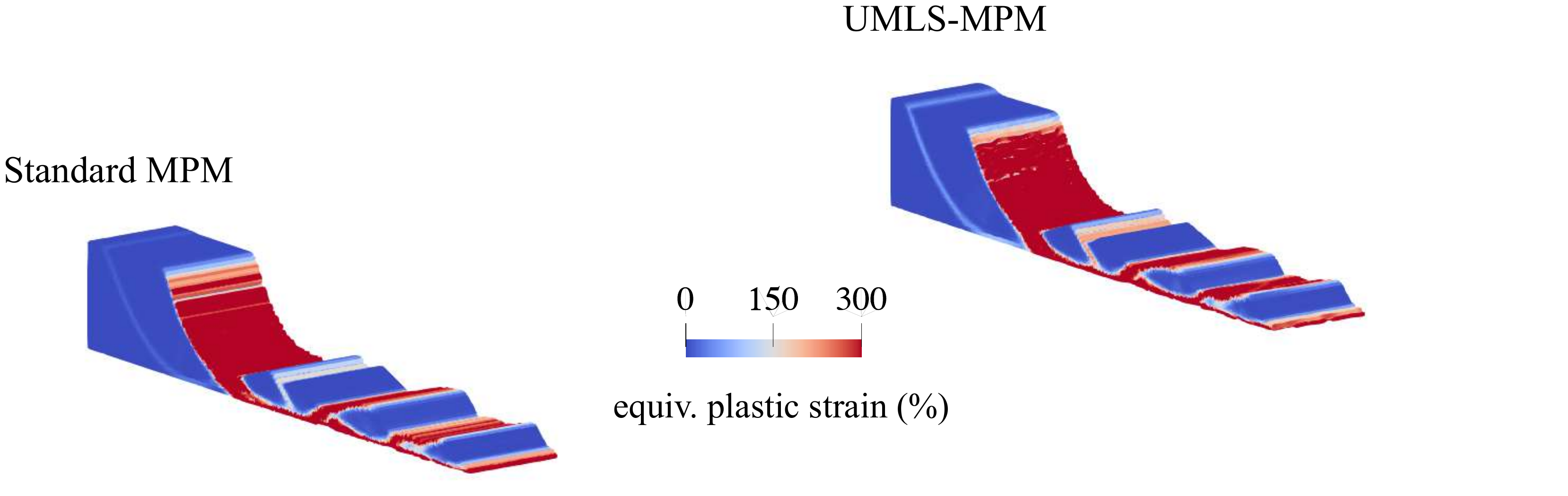}} \\
    \subfloat[$t=5.5$ s]{\includegraphics[width=0.9\textwidth]{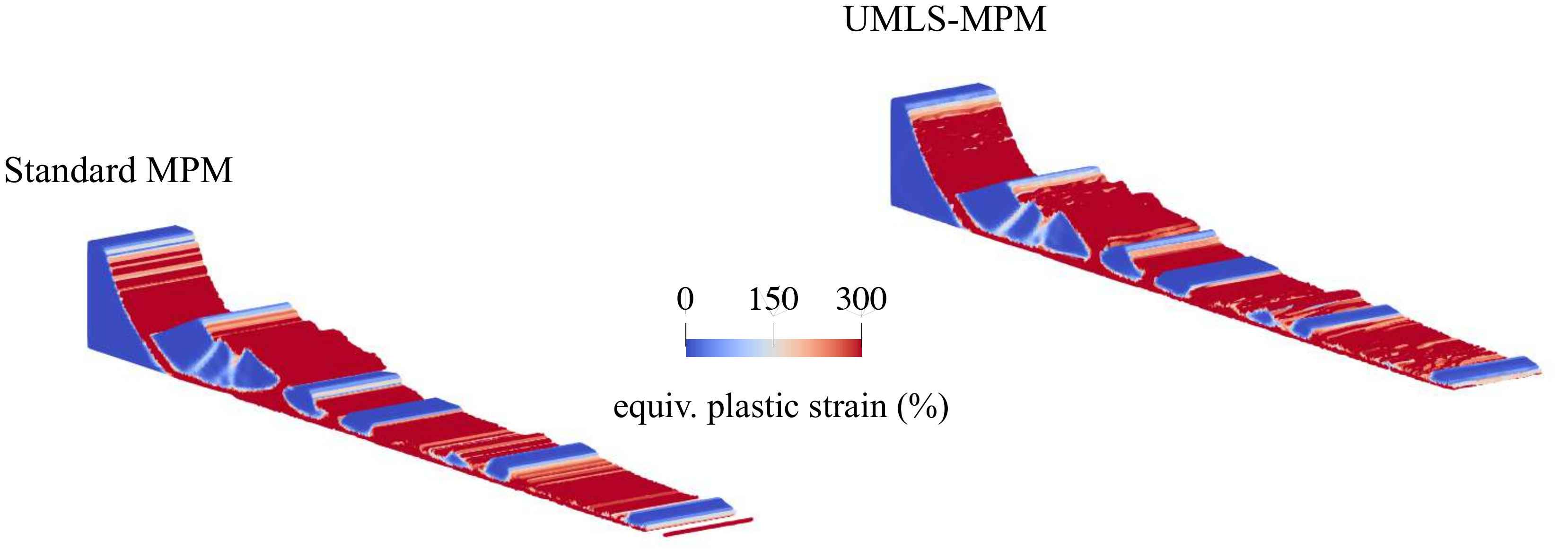}}
    \caption{\revision{Snapshots of the solutions from the \ours\ and the standard MPM with GIMP basis functions in Zhao~\etal~\cite{zhao2023circumventing}. Particles are colored based on the equivalent plastic strain.}}
    \label{fig:3d_slope_strain}
\end{figure}

\begin{figure}[htbp]
    \centering
    \subfloat[$t=1.5$ s]{\includegraphics[width=0.9\textwidth]{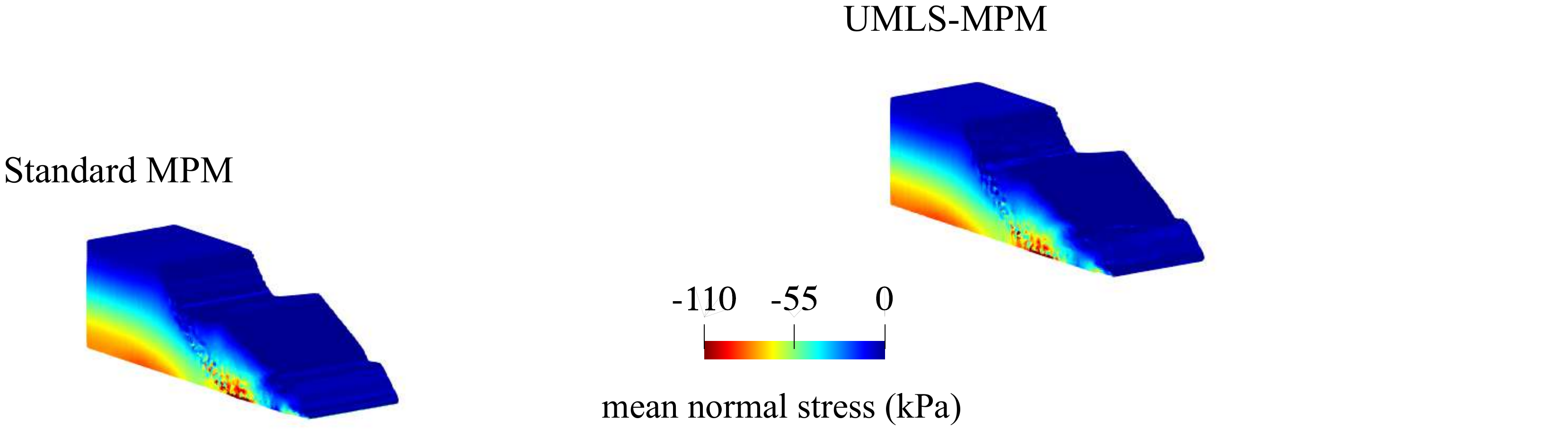}} \\
    \subfloat[$t=2.5$ s]{\includegraphics[width=0.9\textwidth]{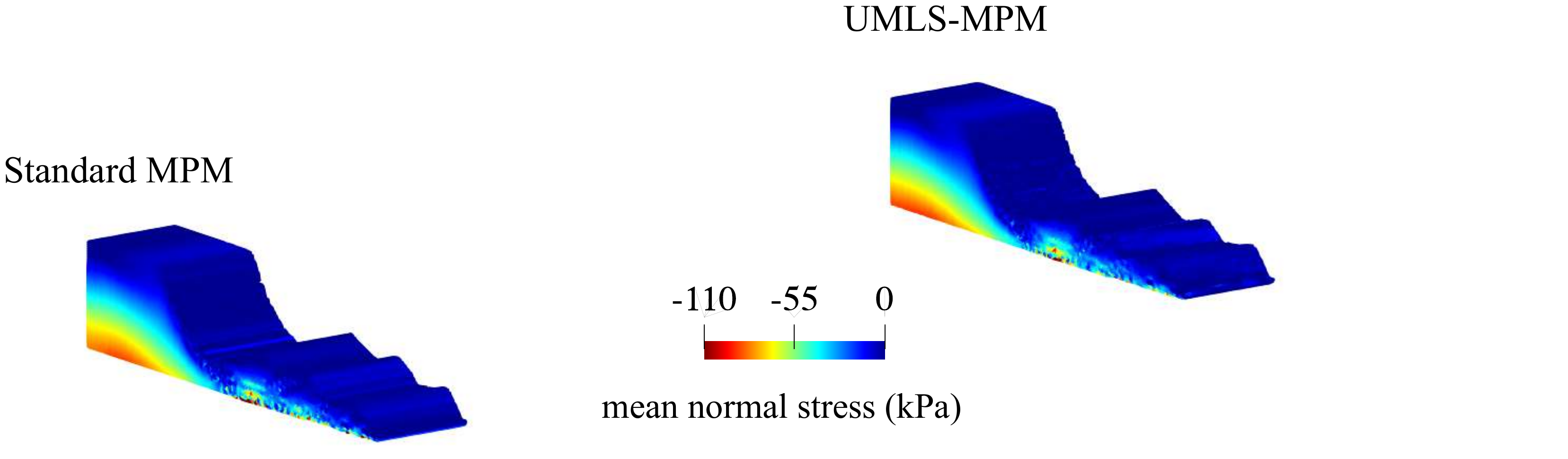}} \\
    \subfloat[$t=3.5$ s]{\includegraphics[width=0.9\textwidth]{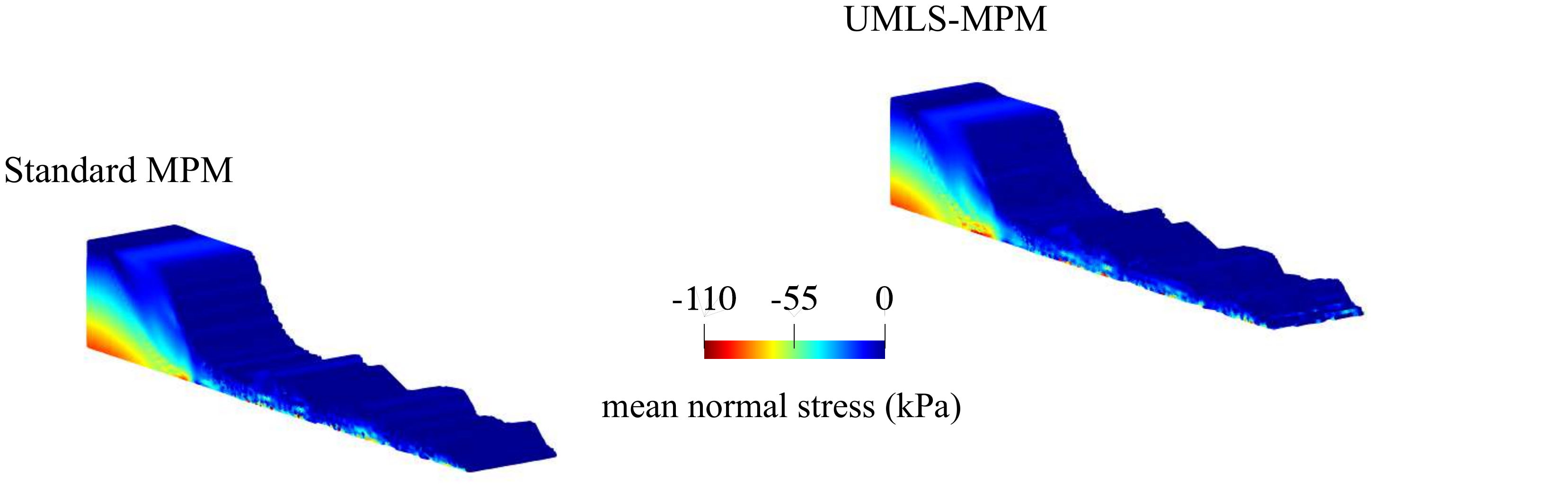}} \\
    \subfloat[$t=5.5$ s]{\includegraphics[width=0.9\textwidth]{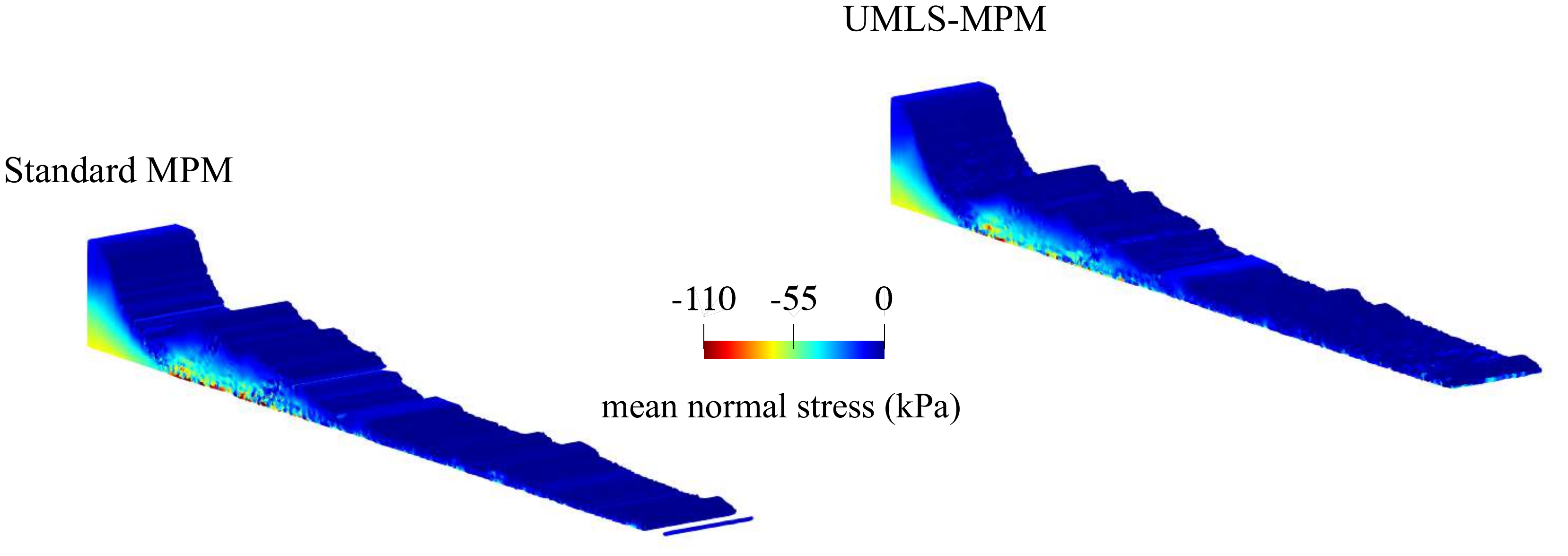}}
    \caption{\revision{Snapshots of the solutions from the \ours\ and the standard MPM with GIMP basis functions in Zhao~\etal~\cite{zhao2023circumventing}. Particles are colored based on the mean normal stress.}}
    \label{fig:3d_slope_stress}
\end{figure}

The space was discretized using Delaunay triangulation with the shortest edge length of 0.2\,m.
The material points were initialized with a spacing of 0.1\,m in each direction, amounting to 311,250 material points in the initial slope region.
Note that the spatial discretization aligns with the one used in \cite{zhao2023circumventing} in terms of both the shortest edge length of the background element and the number of material points.
Also, the $\bar{\bm{F}}$ approach proposed in \cite{zhao2023circumventing} was utilized to circumvent volumetric locking that \ours\ solutions encounter when simulating a large number of particles of incompressible materials.
As a reference to verify the correctness of the proposed formulation, the $\bar{\bm{F}}$ solution in \cite{zhao2023circumventing} was used.

\begin{figure}[htbp]
    \centering
    \includegraphics[width=0.6\textwidth]{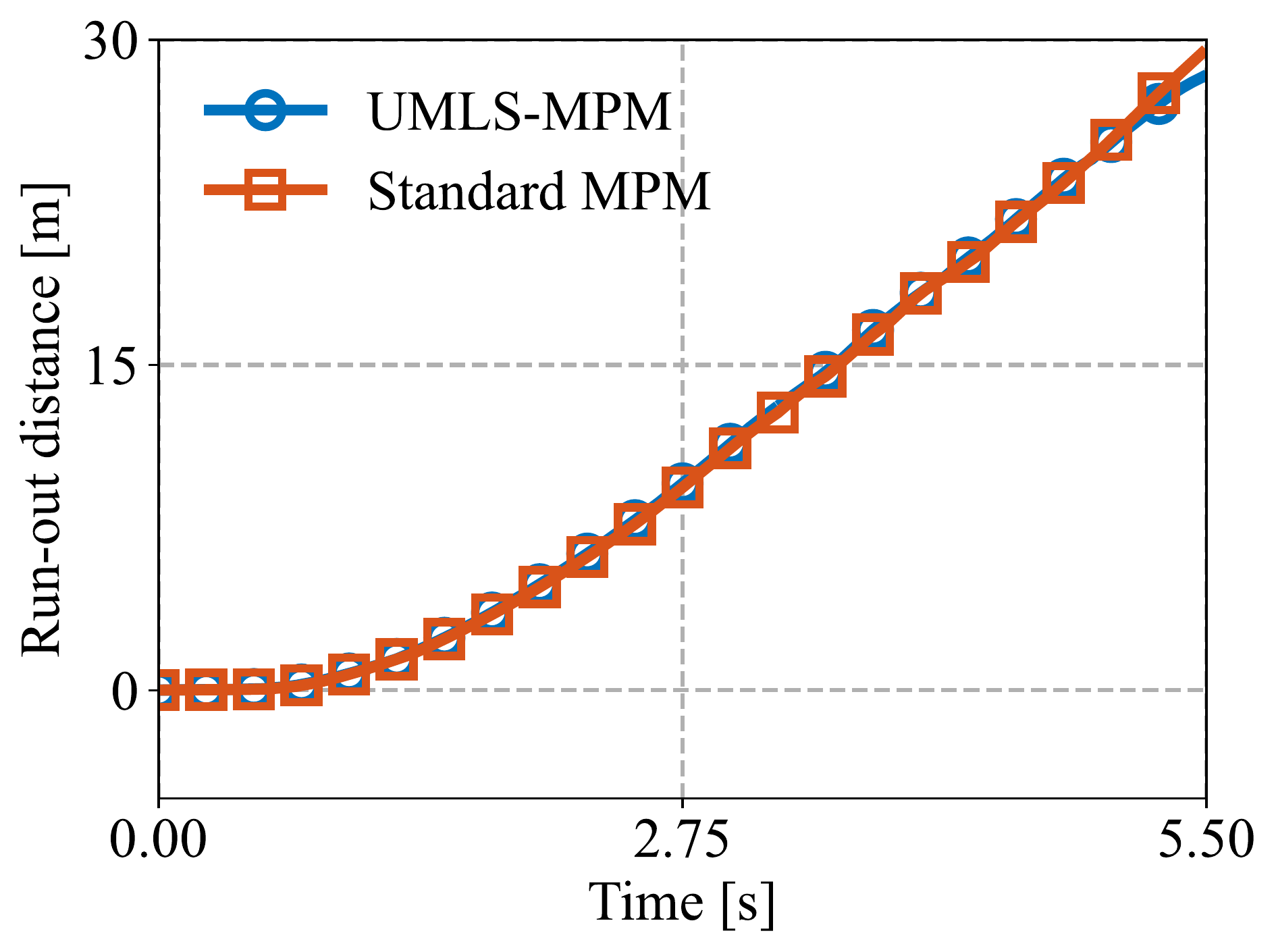}
    \caption{\revision{Time evolutions of the run-out distance from \ours\ and the standard MPM with GIMP basis functions.}}
    \label{fig:3d_slope_runout}
\end{figure}

Figures~\ref{fig:3d_slope_strain} and \ref{fig:3d_slope_stress} show the snapshots of the slope simulated by the standard and \ours, where particles are colored by the equivalent plastic strain and mean normal stress, respectively.
We can see that \ours\ effectively captures the retrogressive failure pattern of slopes made of sensitive clay.
Also, in terms of equivalent plastic strain fields and mean normal stress fields, we observe a strong similarity between the \ours\ solution and the reference solution from \cite{zhao2023circumventing}.

For a further quantitative comparison, Figure~\ref{fig:3d_slope_runout} presents the time evolutions of the run-out distance---a measure of the farthest movement of the sliding mass.
Observe that the distances in the standard and \ours\ solutions are remarkably similar.
Taken together, these findings confirm that the proposed method performs similarly to the standard MPM.

\subsection{3D Elastic Object Expansion in a Spherical Container}

\begin{figure}[htbp]
    \centering
    \includegraphics[width=0.6\textwidth]{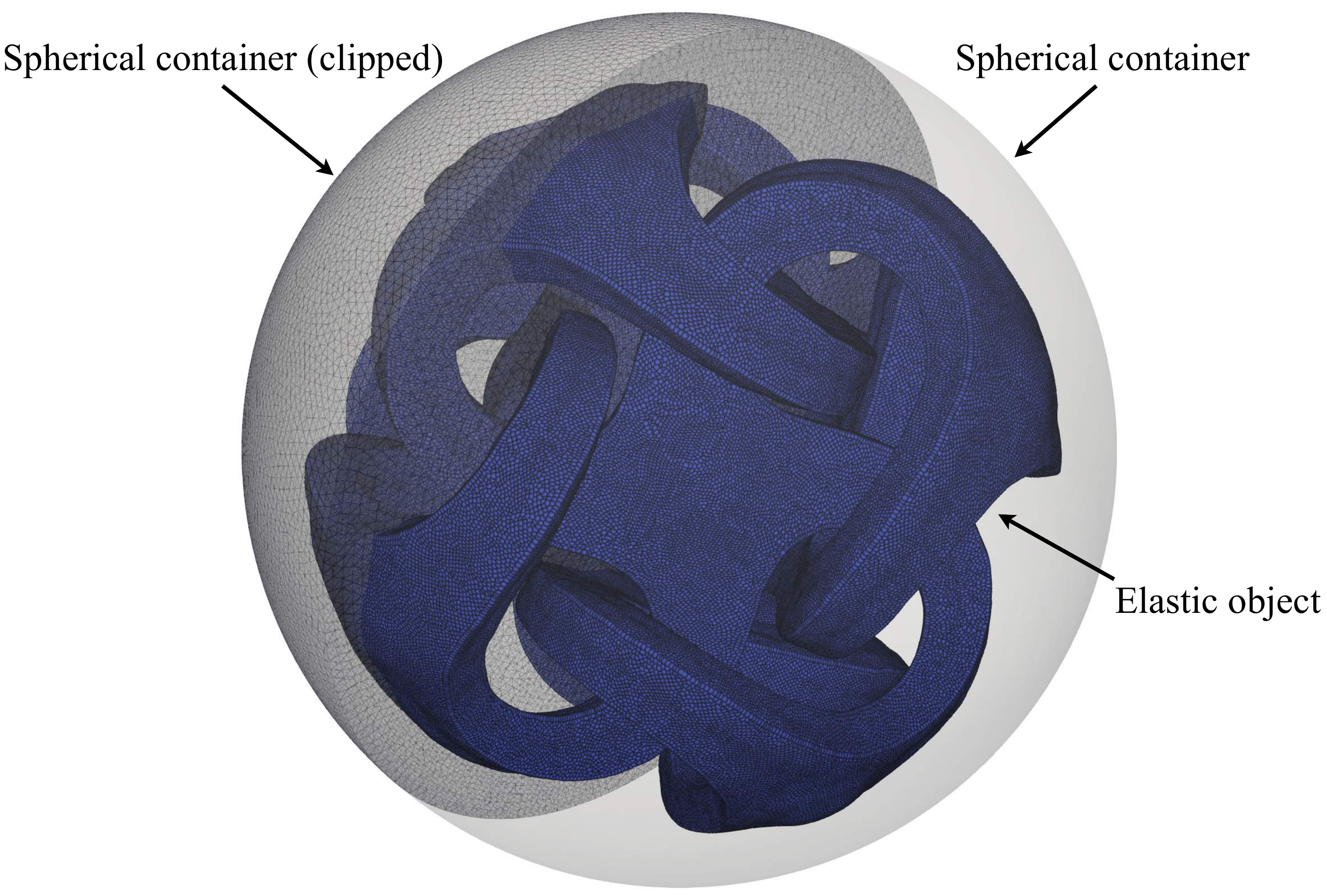}
    \caption{Problem setting of the 3D elastic object expansion in a spherical container.}
    \label{fig:collision_setup}
\end{figure}

\begin{figure}[htbp]
    \centering
    {\includegraphics[width=0.8\textwidth]{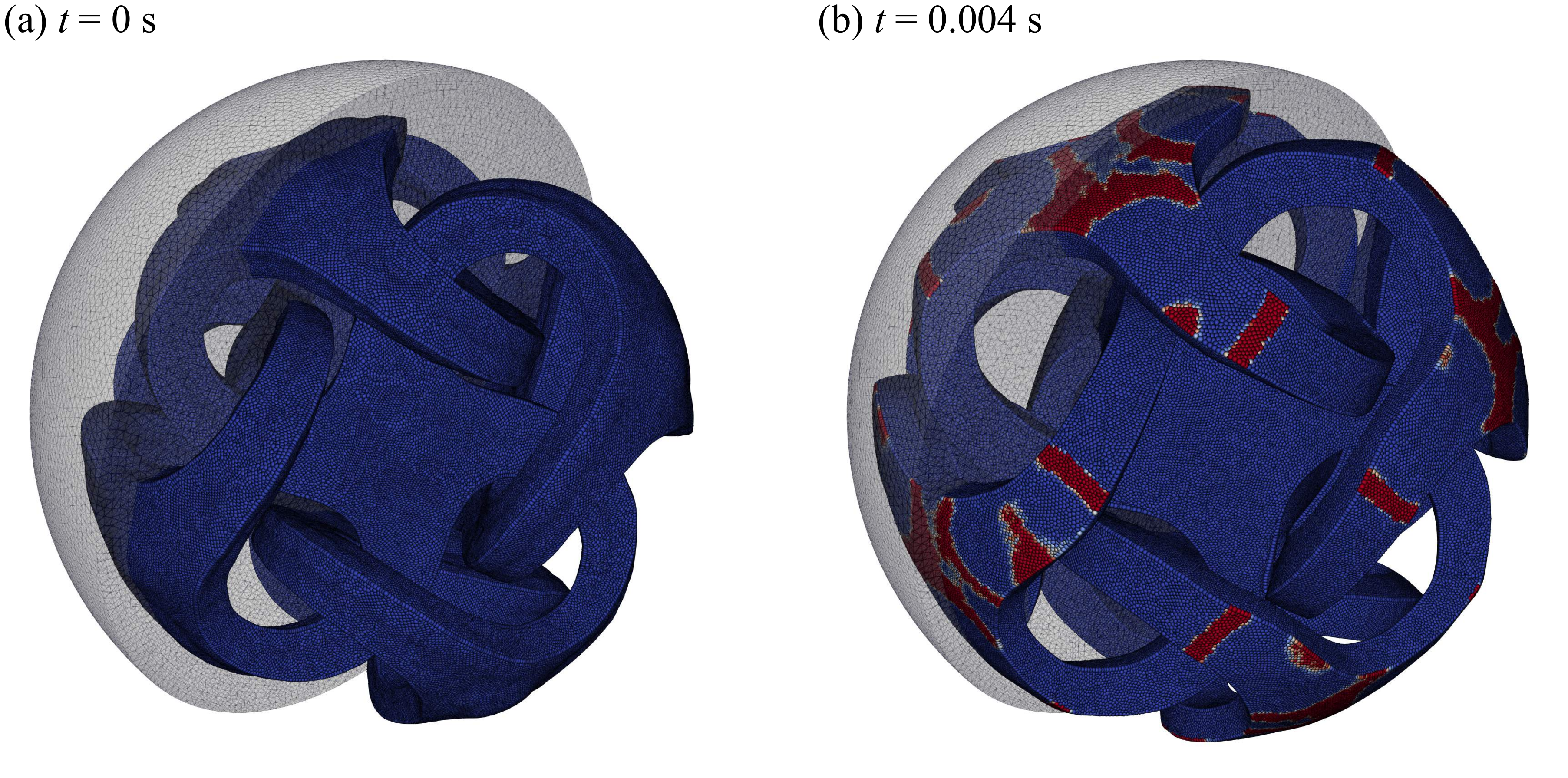}} \\
    {\includegraphics[width=0.8\textwidth]{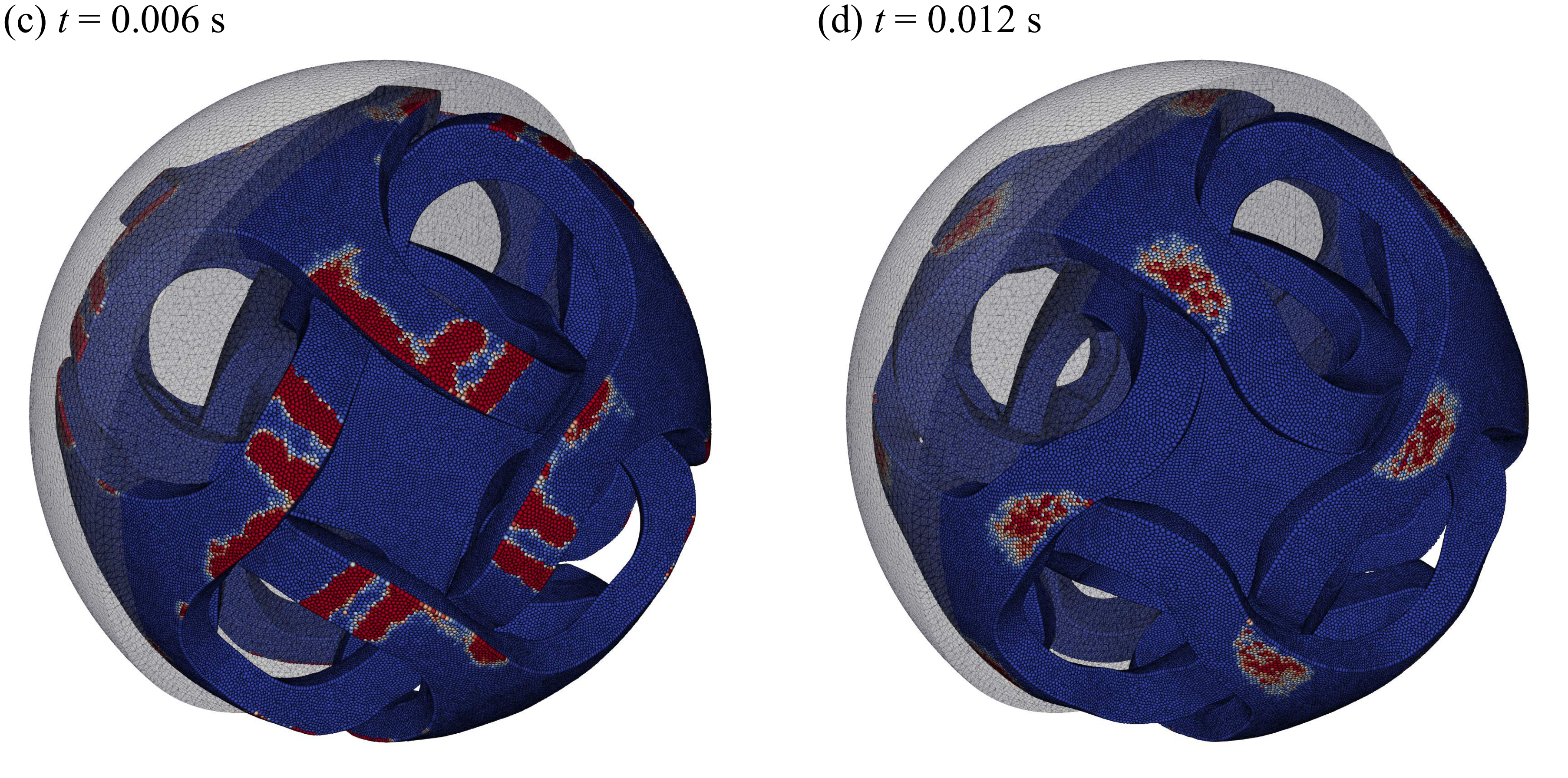}} \\
    {\includegraphics[width=0.8\textwidth]{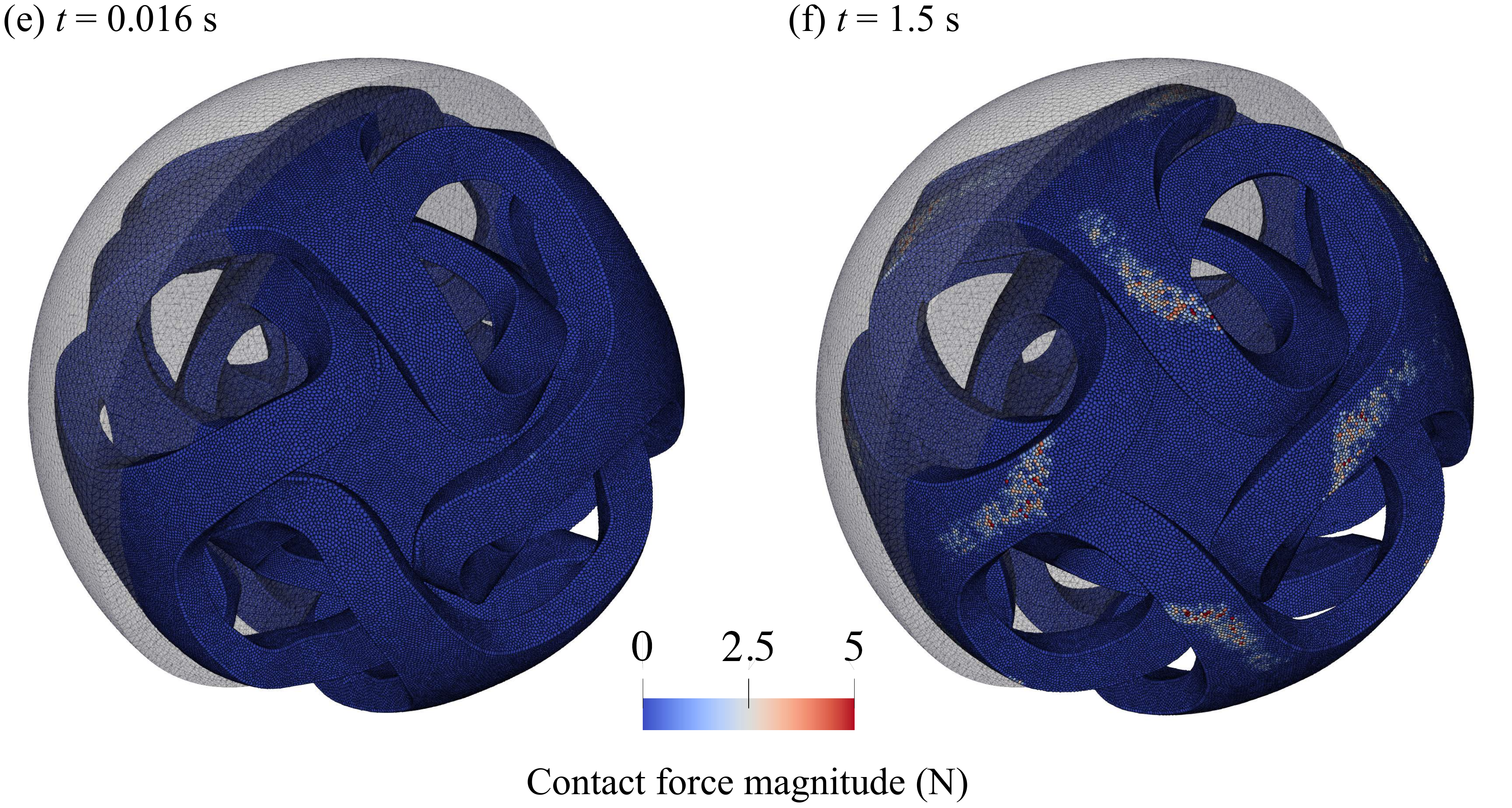}} \\
    \caption{Snapshots from the \ours\ solutions in which particles are colored based on the magnitude of the contact force.}
    \label{fig:collision_snapshots}
\end{figure}

Finally, we examined the performance of \ours\ in problems involving complex boundary geometry.
In this problem, the standard MPM with a structured grid may be challenged to impose conforming boundary conditions.
Hence, a collision between an elastic body with a spherical container was considered and simulated.

The geometry of the problem, as demonstrated in Figure~\ref{fig:collision_setup}, involves an elastic object in the shape of a Metatron, which is located at the center of a spherical container (with a radius of 0.5\,m).
The object is initially compressed isotropically (with an initial deformation gradient of $\tensor{F} = 0.75 \tensor{I}$), storing non-zero elastic potential energy.
At the onset of the simulation, the stored elastic energy is released, causing the object to expand and collide with the spherical container's boundary.
To capture the elastic behavior of the object, a Neohooken elasticity was adopted with a Young's modulus of 3.3\,MPa and a Poisson's ratio of $\nu = 0.49$.
The elastic object was discretized using a significant number of material points (2,392,177) for high-fidelity simulation.
Also, the spherical container was discretized using 2,178,129 tetrahedral elements, each with an average edge length of $h= 0.025$\,m.
Note that to avoid negative kernel values at boundary nodes, an extra layer of elements was added outside the original boundary, as discussed in Section~\ref{sec:verify:continous:kernel}.

To consider the frictional collision between the elastic object and the container boundary, a barrier approach~\cite{zhao2023coupled} was adopted, ensuring that the elastic object does not penetrate the boundary.
Contact forces are applied when the distance between a material point and the boundary is below a specific value $\hat{d}$, which was chosen to be a quarter of $h$ for sufficient accuracy.
Also, a friction coefficient of $\mu = 0.5$ was introduced to stop the sliding of the elastic object in the later stages.
The simulation ran with a time increment of $\Delta t = 6.16 \times 10^{-5}$\,s until $t = 1.5$\,s.

Figure~\ref{fig:collision_snapshots} presents six snapshots simulated with \ours, where the particles are colored based on the magnitude of the contact force.
The dynamic behavior of the object at various stages is well captured, including the initial expansion stage (a), the first collision stage (b--d), the rebounding stage (e), and the final static stage (f).
Overall, the \ours\ effectively handles complex geometry with a conformal discretization, which is critical for simulating a wide range of interactions between deformable objects and complex boundaries.

\section{Conclusion, Limitations, and Future Work}
\label{sec:conclusion}
This study has extended the Moving Least Square Material Point Method to encompass general tessellations within both 2D and 3D meshes. This advance has been achieved through the multiplication of a diminishing function to the MLS sample weights. Analytically proved, the approach ensures continuous kernel reconstruction and provides a sound foundation for MPM on any unstructured mesh types. Several numerical experiments in both 2D and 3D domains have demonstrated the method's effectiveness in achieving high-order convergence and eliminating cell-crossing errors.
However, the proposed method still has limitations.

To ensure the MLS solution does not degenerate, there are quality requirements for the surrounding nodes of the particle, which pose challenges, especially in meshes with sharp resolution transitions or poor quality. Additionally, related to varying mesh resolution is the sample weights: we used a B-spline function with a fixed support radius equaling the maximum edge size to ensure that the first ring of neighbors is reserved in the coarsest area. However, this strategy leads to nearly uniform sample weights in finer regions, blurring the kernel. A potential solution could be to incorporate a sizing field within the mesh to dynamically adjust the sample weight function.

\revision{
Although \ours~has advantages in conforming to irregular geometry, this also brings potential issues and leaves room for further improvements.
}
For example, the lack of surrounding samples on one side when the particle is inside the boundary cell can lead to kernel degeneration. While this issue was mitigated in our experiments by drawing an extra layer of cells, an automatic and algorithmic approach is more appealing.
\revision{
Combining \ours~with regular MPM is also promising, as it leverages the advantage of \ours~in conforming to boundary shapes, as well as the robustness and efficiency of regular MPM for interior domains. For the transition between interior nodes and boundary nodes, methods similar to immersed FEM~\cite{zhang2004immersed,liu2006immersed,li2022immersed} could be effective, as they project information from regular background grids to irregular meshes.
}

\revision{
From a practical standpoint, we need not a higher-order polynomial basis or higher rings of neighbors as long as the current schemes, which have the lowest cost as discussed in~\Secref{sec:method}, provides sufficient accuracy. However, from a theoretical perspective, these two prososals are still interesting. Several challenges can be noted:
1. Changing to higher rings of neighbors will require a different solution for the diminishing function~\Eqref{eq:invD:diminish} to accommodate the new discrete change of active sets.
2. The higher-order polynomial basis will result in a different MLS solution compared to the linear solution~\eqref{eq:mls:linear:ours:recon}. Additionally, since the propositions in \Appref{appdx:conservation:proof} are based on the linear solution, it remains an open question whether the higher-order solution has a continuous kernel and gradient, even with a diminishing function.
3. The higher-order polynomial basis will necessitate a more delicate selection of quadrature point locations. Unlike higher-order FEM simulations, where the quadrature points are usually well-designed and fixed~\cite{hughes2012finite}, in MPM, these points may need to be adjusted with the particle locations, requiring extra design and attention.
}

From an application standpoint, exploring the integration between the Material Point Method (MPM) and gas or fluid simulations via the Finite Volume Method (FVM) presents significant potential \cite{cao2022efficient,baumgarten2021coupled}.
\revision{
For simulating granular materials in realistic and irregular container shapes, combining \ours~with the Discrete Element Method (DEM)~\cite{cook2002discrete} is also an interesting direction. Moreover, object contact detection and handling~\cite{li2023contact} is crucial to prevent artificial penetration and sticking commonly seen in MPM~\cite{jiang2016material}.
}
In the field of scientific machine learning, recent advances have demonstrated learning MPM or other mesh-based simulations using graph neural networks, accelerating the inferences \cite{pfaff2020learning,li2023mpmnet,cao2023efficient}.
Our kernel construction suggests a potentially novel learning paradigm for MPM on unstructured meshes, similar to embedding both kernel and mesh information into the network's channel \cite{gao2022finite,li2023finite}.


\newpage
\bibliographystyle{plainnat} 
\bibliography{references}

\newpage
\appendix
\revision{
\section{Pipeline for Rapid Cell Search}
\label{appdx:pipeline:hash:search}
Fast determination of which cell contains a specific particle is crucial for the efficiency of \ours. We propose a hash grid-based method to accelerate this process. The pipeline consists of two steps: (1) building the hash grid connectivity table, a dictionary mapping the hash grid as a key to all of its touching cell indices, as detailed in Algorithm~\ref{alg:build:hash2cell}, and (2) the online search for determining which cell contains a given particle, as described in Algorithm~\ref{alg:find:containing:cell}.
}

\begin{algorithm}[t]
    \caption{Build Hash to Adjacent Cell}
    \begin{algorithmic}[1]
        \State \textbf{Input:} Mesh node positions: \texttt{pos}, mesh cells: \texttt{cell}, hash grid size: \texttt{dx}
        \State \textbf{Output:} Hash to adjacent cell connectivity: \texttt{hash2cell}
        
        \State \textbf{Calculate the bounding box of the whole mesh}
        \Statex \hspace{1cm} $\texttt{min\_pos} \gets \texttt{min(pos, axis=0)}$
        \Statex \hspace{1cm} $\texttt{max\_pos} \gets \texttt{max(pos, axis=0)}$
               
        \State \textbf{Initialize \texttt{hash2cell} as an empty dictionary}
        \Statex \hspace{1cm} $\texttt{hash2cell} \gets \{\}$
        
        \State \textbf{For each cell, find its bounding box, determine the spatial hash grids it touches, and append the cell to those hash grids}
        \For{\texttt{cell\_idx} in \texttt{range(len(cell))}}
            \State \hspace{1cm} $\texttt{cell\_min} \gets \texttt{min(pos[cell[cell\_idx]], axis=0)}$
            \State \hspace{1cm} $\texttt{cell\_max} \gets \texttt{max(pos[cell[cell\_idx]], axis=0)}$
            \State \hspace{1cm} $\texttt{min\_idx} \gets \texttt{floor((cell\_min - min\_pos) / dx)}$
            \State \hspace{1cm} $\texttt{max\_idx} \gets \texttt{ceil((cell\_max - min\_pos) / dx)}$
            
            \State \hspace{1cm} $\texttt{range} \gets \texttt{indices(max\_idx - min\_idx)}$
            \State \hspace{1cm} $\texttt{candidates} \gets \texttt{range.reshape(dim, -1).T + min\_idx}$
            \For{\texttt{c} in \texttt{candidates}}
                \State \hspace{2cm} $\texttt{hash2cell[c].append(cell\_idx)}$
            \EndFor
        \EndFor
    \end{algorithmic}
    \label{alg:build:hash2cell}
\end{algorithm}

\begin{algorithm}[t]
    \caption{Find Containing Cell}
    \begin{algorithmic}[1]
        \State \textbf{Input:} Mesh node positions: \texttt{pos}, mesh cells: \texttt{cell}, Mesh minimum position: $\texttt{min\_pos}$, hash grid size: \texttt{dx}, Hash to adjacent cell connectivity: \texttt{hash2cell}, particle position: \texttt{xp}
        \State \textbf{Output:} Containing cell index or \texttt{None}

        \State \textbf{Get the hash spatial index for the particle}
        \Statex \hspace{1cm} \texttt{hash\_idx} $\gets$ \texttt{floor((xp-min\_pos) / dx)}
        
        \State \textbf{Get all adjacent cells to this hash grid}
        \Statex \hspace{1cm} \texttt{candidates} $\gets$ \texttt{hash2cell[hash\_idx]}
        
        \State \textbf{For each candidate cell, check if the particle is inside}
        \For{\texttt{c} in \texttt{candidates}}
            \State \texttt{e} $\gets$ \texttt{cell[c]}
            \State \texttt{vs} $\gets$ \texttt{pos[e]}
            \State \# get the bounding box of the cell
            \State \texttt{e\_min} $\gets$ \texttt{min(vs, axis=0)}
            \State \texttt{e\_max} $\gets$ \texttt{max(vs, axis=0)}
            \State \# quick filter using bounding box
            \If{\texttt{all(xp >= e\_min) and all(xp <= e\_max)}}
                \State \# check barycentric coordinates
                \State \texttt{bc} $\gets$ \texttt{barycentric\_coord(vs, xp)}
                \If{\texttt{all(bc >= 0)}}
                    \State \textbf{return} \texttt{c}
                \EndIf
            \EndIf
        \EndFor
        \State \# If no candidate cells contain \texttt{xp}
        \State \textbf{return} \texttt{None}
    \end{algorithmic}
    \label{alg:find:containing:cell}
\end{algorithm}

\section{Proofs for the Continuous Reconstructions}
\label{appdx:conservation:proof}

For conciseness, we drop the subscripts ${p}$ in the following proofs.
We start by assuming there exists a smooth, locally diminishing function $\eta$ for the nodes added or removed from the set of nearby nodes $\gN^1$ when a particle crosses the boundary of a cell. Under this assumption, we can prove that our kernel value and gradient estimation is continuous across the boundary. We present the proof in 2D when a particle crosses an edge; the extension to 3D and other crossing cases is straightforward. Finally, we prove that \Eqref{eq:invD:diminish} satisfies the forementioned assumption.

\begin{prop}
    Our kernel value and gradient estimation is continuous across cell boundaries.
\end{prop}

\begin{proof}
    Let $\gN_{o,n}^1$ be the sets of nearby nodes before/after the particle $p$ crosses the common edge between the old/new cells $\gN_{o,n}^0$. Here, the subscripts ${o,n}$ denote the old/new cell, respectively, and the superscripts ${0,1}$ indicate the ring-0/1 neighbors of the cell, respectively. Let $\vx^{o,n}$ be the position of particle $p$ before/after the crossing and $||\vx^{n}-\vx^{o}|| = \gO(\epsilon)$.
    Define the common node set $\gN_{c}^1 = \gN_{o}^1 \cap \gN_{n}^1$, the added node set $\gN_{a}^1 = \gN_{n}^1 \setminus \gN_{c}^1$, and the removed node set $\gN_{r}^1 = \gN_{o}^1 \setminus \gN_{c}^1$.
    Since $\eta$ is locally diminishing for $v \in \gN_{a,r}^1$, we have a positive value $K_1$ such that $\eta = \gO(K_1 \epsilon) = \gO(\epsilon)$.
    The pertubation for the assembled matrix $\mM$ before/after the particle $p$ crosses an edge is
    \begin{equation}
        \delta \mM = \sum_{v \in \gN_{c}^1} \delta (\eta d \vp \vp^T) + \sum_{v \in \gN_{a}^1} \eta d \vp \vp^T - \sum_{v \in \gN_{r}^1} \eta d \vp \vp^T,
    \end{equation}
    where the first term is continuous by construction since every factor is smooth; \ie, $ ||\delta (\eta d \vp \vp^T)|| = \gO(\epsilon)$. For the second and third terms, since $\eta = \gO(\epsilon)$, we have
    \begin{equation}
        \begin{aligned}
            ||\delta \mM|| & \leq \sum_{v \in \gN_{c}^1} ||\delta (\eta d \vp \vp^T)|| + \sum_{v \in \gN_{a}^1} ||\eta d \vp \vp^T|| + \sum_{v \in \gN_{r}^1} ||\eta d \vp \vp^T|| \\
                           & \leq \biggl[|\gN_{c}^1| + \left(|\gN_{a}^1| + |\gN_{r}^1|\right) \max_{v \in \gN^1_{a,r}} ||d \vp \vp^T|| \biggr] \gO(\epsilon)                                          \\
                           & = \gO(|\gN^1| h^2 \epsilon)                                                                                                                                                   \\
                           & = \gO(\epsilon).                                                                                                                                                             \\
        \end{aligned}
    \end{equation}
    Here, as long as the mesh has a reasonably good quality, $|\gN^1|$ is finite and small; \ie, there is a finite and small amount of ring-1 neighbors. Also, $h$, a constant, is the support radius of the kernel, outside of which the weight is zero. In all, both $|\gN^1|$ and $h$ can be omitted in the analysis.

    The perturbation of the inverse matrix is given by
    \begin{equation}
        \begin{aligned}
            ||\delta \mM^{-1}|| & = ||(\mM + \delta \mM)^{-1} - \mM^{-1}||                                         \\
                                & = ||\mM^{-1} - \mM^{-1} \delta \mM \mM^{-1} + \gO(||\delta \mM||^2) - \mM^{-1}|| \\
                                & = ||\mM^{-1} \delta \mM \mM^{-1} + \gO(\epsilon^2)||                             \\
                                & \leq ||\mM^{-1} \delta \mM \mM^{-1}|| + \gO(\epsilon^2)                          \\
                                & \leq ||\mM^{-1}||^2 \cdot ||\delta \mM|| + \gO(\epsilon^2)                       \\
                                & = \frac{||\delta \mM||}{\sigma(\mM)^2_{\min}} + \gO(\epsilon^2)                  \\
                                & = \gO\left(\frac{\epsilon}{\sigma(\mM)^2_{\min}} \right) + \gO(\epsilon^2)       \\
                                & = \gO\left(\frac{\epsilon}{\sigma(\mM)^2_{\min}} \right),
        \end{aligned}
    \end{equation}
    where $\sigma(\mM)_{\min}$ is the minimum singular value of $\mM$.

    Similarly, for the perturbation in the assembled vector $\mB \vu = \mP \mD \vu$ before/after the particle crossing is
    \begin{equation}
        \begin{aligned}
            ||\delta \left(\mP \mD \vu\right) || & = ||\sum_{v \in \gN_{c}^1} \delta (\eta d u \vp) + \sum_{v \in \gN_{a}^1} \eta d u \vp - \sum_{v \in \gN_{r}^1} \eta d u \vp||            \\
                                                   & \leq \sum_{v \in \gN_{c}^1} ||\delta (\eta d u \vp)|| + \sum_{v \in \gN_{a}^1} ||\eta d u \vp|| + \sum_{v \in \gN_{r}^1} ||\eta d u \vp|| \\
                                                   & \leq \left[|\gN_{c}^1| + \left(|\gN_{a}^1| + |\gN_{r}^1|\right) \max_{v \in \gN^1_{a,r}} ||d u \vp|| \right] \gO(\epsilon)                                  \\
                                                   & = \gO(|\gN^1| h \epsilon)                                                                                                                                         \\
                                                   & = \gO(\epsilon).
        \end{aligned}
    \end{equation}
    Furthermore, we can establish the following bound for the assembled vector $\mP \mD \vu$:
    \begin{equation}
        \begin{aligned}
            ||\mP \mD \vu|| & = ||\sum_{v \in \gN^1} \eta d u \vp||                  \\
                              & \leq |\gN^1| \cdot \max_{v \in \gN^1} ||\eta d u \vp|| \\
                              & = \gO(|\gN^1| h)                                               \\
                              & = \gO(1).                                                      \\
        \end{aligned}
    \end{equation}
    Finally, the perturbation for $\left[ \hat{u}, \nabla \hat{u}^T \right]^T$ from \Eqref{eq:mls:linear:ours:recon} is
    \begin{equation}
        \begin{aligned}
            \left[ \hat{u}, \nabla \hat{u}^T \right]^T 
            & = ||\delta (\mM^{-1} \mB \vu)||                                                                       \\
            & = ||\delta (\mM^{-1} \mP \mD \vu)||                                                                       \\
                                                           & = ||\delta \mM^{-1} \mP \mD \vu + \mM^{-1} \delta \left( \mP \mD \vu\right)||                           \\
                                                           & \leq ||\delta \mM^{-1} \mP \mD \vu|| + ||\mM^{-1} \delta \left( \mP \mD \vu\right)||                    \\
                                                           & \leq ||\delta \mM^{-1}|| \cdot ||\mP \mD \vu|| + ||\mM^{-1}|| \cdot ||\delta \left(\mP \mD \vu\right)|| \\
                                                           & = \gO\left(\left(\frac{1}{\sigma(\mM)^2_{\min}} + \frac{1}{\sigma(\mM)_{\min}}\right)\epsilon\right).       \\
        \end{aligned}
    \end{equation}

    In the incomplete singular value decomposition of $\mM$, the singular values will always be non-negative. And if the surrounding nodes are not degenerate, the minimum singular value $\sigma(\mM)_{\min}$ will always be positive and the condition number of $\mM$ is bounded. Therefore, as long as the mesh is of reasonably good quality, both the function value and gradient estimation is $\gC^0$ across the boundary.
\end{proof}

\begin{prop}
    The function $\eta_i$ in \Eqref{eq:invD:diminish} is locally diminishing for $\forall i \in \gN_{a,r}^1$.
\end{prop}

\begin{proof}
    Formally, we need to prove that for any $i \in \gN_{a,r}^1$, when $\vx$ is crossing the edge of a triangle and $||\vx^{n}-\vx^{o}|| = \gO(\epsilon)$, the smoothing function $\eta_i = \gO(\epsilon)$.

    Denote the edge that the particle is crossing as $e$ and the portion of $||\vx^{n}-\vx^{o}||$ in the new/old cell as $L^{n,o}$. Trivially,
    \begin{equation}
        \begin{aligned}
            L^{n,o} & \leq L^{n} + L^{o}        \\
                    & = ||\vx^{n}-\vx^{o}|| \\
                    & = \gO(\epsilon).
        \end{aligned}
    \end{equation}
    Then, let the far-away node not on the edge but in the new/old cell be $i_{\text{far}}$ (\ie, $i_{\text{far}}^{n,o}\notin e \,\land\, i_{\text{far}}^{n,o} \in \gN_{o,n}^0$) and the height from a node $i$ to an edge $e$ be $H(i, e)$. Since the height is orthogonal to the edge, we have $H(\vx^{n,o}, e) \leq L^{n,o} = \gO(\epsilon)$. Consider the barycentric coordinate contributed by the far-away node, in the new/old cell respectively, for $\vx$:
    \begin{equation}
        \begin{aligned}
            B^{n,o}_{i_{\text{far}}} & = \frac{H(\vx^{n,o}, e) \cdot ||e||}{H(i_{\text{far}}^{n,o}, e) \cdot ||e||} \\
                              & = \frac{H(\vx^{n,o}, e)}{H(i_{\text{far}}^{n,o}, e)}                         \\
                              & = \gO(\frac{\epsilon}{H(i_{\text{far}}^{n,o}, e)})                             \\
                              & = \gO(\epsilon).
        \end{aligned}
        \label{eq:bvfarno}
    \end{equation}

\begin{figure}[tb]
    \centering
    \includegraphics[width=0.9\textwidth]{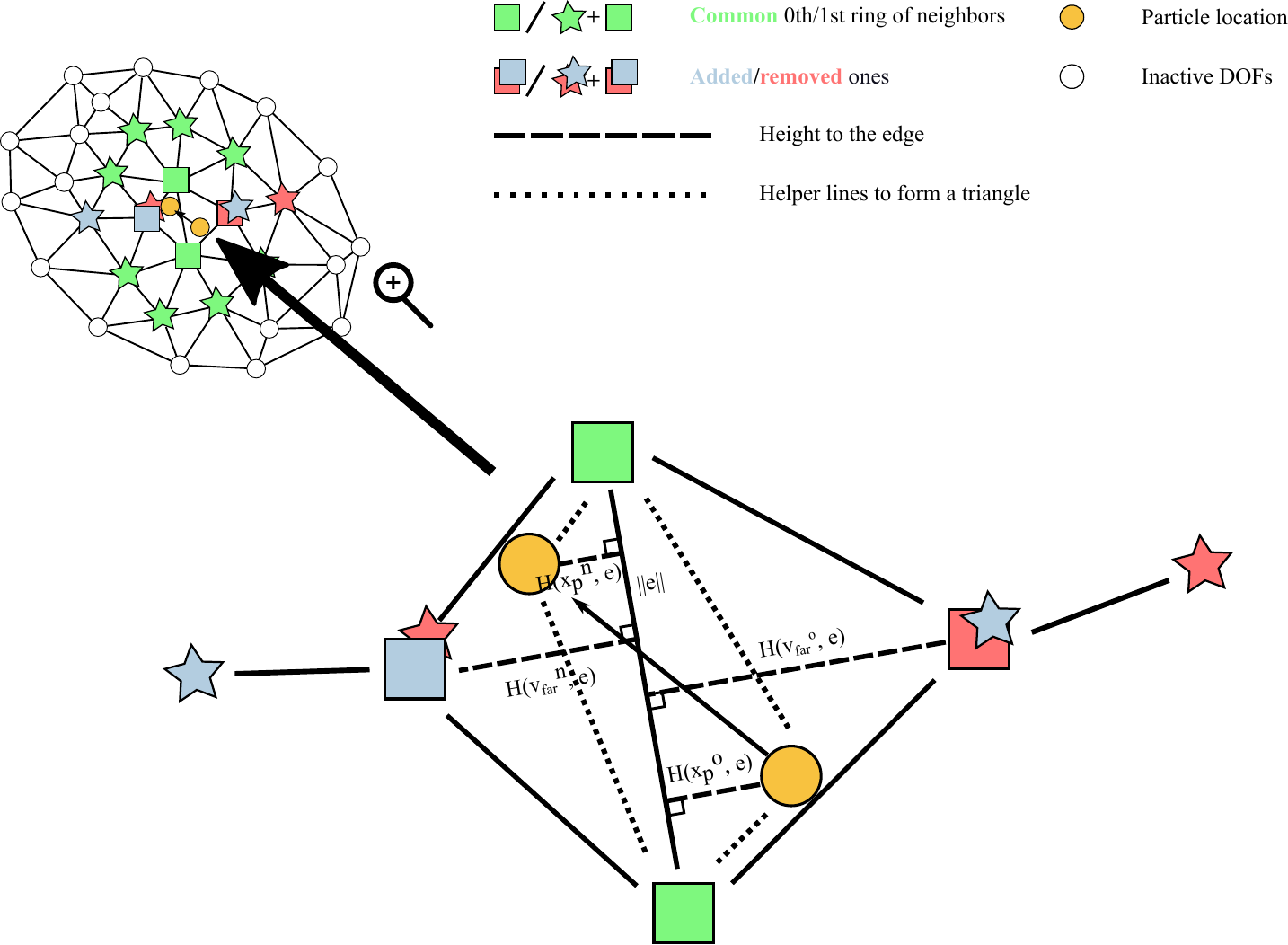}
        \caption{A visual representation of the notations used to prove that $\eta$ diminishes locally, as described in (\Eqref{eq:bvfarno}), for every vertex $i$ within the first ring of neighbors, $\gN^1_{a,r}$. The dashed line denotes the perpendicular height from a given position to the shared edge. Dotted lines are drawn to construct a triangle between the point $\vx$ and the edge, facilitating the computation of the barycentric coordinates.}
        \label{fig:invD:diminish}
\end{figure}S

    Finally, if a node is added/removed during the particle crossing (\ie, $i \in \gN_{a,r}^1$), this means that $i$ is only connected to the far-away nodes $i_{\text{far}}^{n,o}$ but not to the edge $e$; \ie, $\emA_{i,i_{\text{far}}^{n,o}} = 1, \forall i \in \gN^1_{a,r}$, otherwise $\emA_{i, n} = 0, \forall n \in e \land \forall i \in \gN^1_{a,r}$. Hence,
    \begin{equation}
        \begin{aligned}
            \eta & = \sum_{n \in \gN^0} B_n \emA_{i, n}                                        \\
                   & = B^{n,o}_{i_{\text{far}}} \emA_{i,i_{\text{far}}^{n,o}} + \sum_{n \in e} B_n \emA_{i, n} \\
                   & = B^{n,o}_{i_{\text{far}}} \cdot 1 + \sum_{n \in e} B_n \cdot 0                   \\
                   & = B^{n,o}_{i_{\text{far}}}                                                        \\
                   & = \gO(\epsilon), \quad \forall i \in \gN^1_{a,r}.
        \end{aligned}
    \end{equation}
    This concludes the proof.
\end{proof}

\section{Settings and Analytical Solutions for the Verification Experiments of the Continuous Reconstructions}
\label{appdx:kernel:verify:misc}

This section presents the detailed setup and analytical solutions for the 1D verification experiments on a uniform 1D mesh in \Secref{sec:verify:continous:kernel}. The kernel value is denoted as $f$ and the gradient estimation is denoted as $g$, respectively.
For the uniform 1D mesh, each cell has a length of $1$, and the unit support length for the B-spline used for the sample weights is also $1$. The analytical solution for the uniform 1D mesh, obtained using \texttt{Mathematica 2023}, is as follows:
\begin{equation}
    \begin{aligned}
        \begin{array}{cc}
            f= &
            \begin{cases}
                \frac{0.25 (0.5\, -x)^2 x}{x^5-6x^4+13.5 x^3-13.75 x^2+6.3125x-0.5625},                                               & 0.5<x\leq 1      \\
                \frac{-x^6+5x^5-9.5 x^4+8.5 x^3-3.3125 x^2+0.3125 x+0.0625}{-x^5+4x^4-5.5 x^3+3.25 x^2-1.3125x+1.0625},               & 1<x\leq1.5       \\
                \frac{x^6-12x^5+58.5 x^4-148.25 x^3+205.563 x^2-147.063 x+42.25}{x^5-11x^4+47.5 x^3-100.25x^2+103.313 x-41.125},      & 1.5<x\leq2       \\
                \frac{x^6-12x^5+58.5 x^4-147.75 x^3+202.563 x^2-141.438 x+39}{-x^5+9x^4-31.5 x^3+53.75x^2-45.3125 x+16.125},          & 2<x\leq2.5       \\
                \frac{x^6-19x^5+149.5 x^4-623.5 x^3+1453.31 x^2-1794.19 x+915.687}{-x^5+16x^4-101.5 x^3+318.75x^2-495.313 x+304.188}, & 2.5<x\leq 3      \\
                \frac{-0.25 x^3+2.75 x^2-10.0625 x+12.25}{-x^5+14x^4-77.5 x^3+212.25x^2-288.313 x+156.688},                           & 3<x\leq3.5       \\
                0,                                                                                                                    & \text{Otherwise}
            \end{cases}
            \\
        \end{array}
    \end{aligned}
\end{equation}

\begin{equation}
    \begin{aligned}
        \begin{array}{cc}
            g= &
            \begin{cases}
                -\frac{1(0.5\, -x)^2 x (x-2)}{x^5-6x^4+13.5 x^3-13.75x^2+6.3125 x-0.5625} ,                              & 0.5<x\leq 1      \\
                \frac{-x^5+5x^4-10.5 x^3+11.5 x^2-5.8125 x+1.0625}{-x^5+4x^4-5.5x^3+3.25 x^2-1.3125 x+1.0625},            & 1<x\leq1.5       \\
                \frac{-x^5+9x^4-33.5 x^3+65.75 x^2-67.3125 x+27.625}{-x^5+11x^4-47.5 x^3+100.25 x^2-103.313x+41.125} ,    & 1.5<x\leq2       \\
                \frac{-x^5+11x^4-49.5 x^3+112.25 x^2-125.313 x+53.625}{x^5-9x^4+31.5 x^3-53.75 x^2+45.3125x-16.125} ,     & 2<x\leq2.5       \\
                \frac{-x^5+15x^4-90.5 x^3+274.5 x^2-417.813 x+254.188}{x^5-16x^4+101.5 x^3-318.75 x^2+495.313x-304.188} , & 2.5<x\leq 3      \\
                \frac{x^4-13x^3+62.25 x^2-129.5 x+98}{-x^5+14x^4-77.5 x^3+212.25 x^2-288.313 x+156.688} ,                 & 3<x\leq3.5       \\
                0,                                                                                                        & \text{Otherwise}
            \end{cases}
            \\
        \end{array}
    \end{aligned}
\end{equation}

\section{Conservation of the Linear and Affine Momentum When Combined With Affine Particle-in-Cell}
\label{appdx:apic:conservation}

Since \ours\ by construction generates a kernel that is the partition of unity and conserves the linear basis \cite{levin1998approximation}, \ie,
\begin{align*}
    \sum_i w_{p, i}^{n}                     & = 1,       \\
    \sum_i w_{p, i}^{n} \vx_i^n             & = \vx_p^n, \\
    \sum_i w_{p, i}^{n} (\vx_i^n - \vx_p^n) & = 0,
\end{align*}
then the system's total linear and angular momentum will be conserved when combined with APIC. A simple introduction to APIC is given here for the sake of completeness, while the detailed proof can found in the supplementary document of the original APIC paper \cite{jiang2015affine}.

In APIC, mass $m_p$, position $\vx_p$, velocity $\vv_p$, and an affine matrix $\mB_p = \sum_i w_{p, i} v_i (\vx_i-\vx_p)^T$ are stored and tracked on particles. Then,
\begin{definition}
    The total linear momentum on grids is
    $$
        \vp^{tot}_i = \sum_i m_i \vv_i.
    $$
\end{definition}
\begin{definition}
    The total linear momentum on particles is
    $$
        \vp^{tot}_p = \sum_p m_p \vv_p.
    $$
\end{definition}
\begin{definition}
    The total angular momentum on grids is
    $$
        \mI^{tot}_{i} = \sum_i \vx_i \times m_i \vv_i.
    $$
\end{definition}
\begin{definition}
    The total angular momentum on particles is
    $$
        \mI^{tot}_{p} = \sum_p \vx_p \times m_p \vv_p + \sum_p m_p (\mB_p)^T : \mathbb{\epsilon},
    $$
\end{definition}
\noindent where $\mathbb{\epsilon}$ is the Levi-Civita permutation tensor, and for any matrix $\mA$, the contraction $\mA : \mathbb{\epsilon} = \sum_{\alpha\beta} A_{\alpha\beta} \mathbb{\epsilon}_{\alpha\beta\gamma}$, which is usually used to transition from a cross product into the tensor product $\vu \times \vv = (\vv \vu^T)^T : \mathbb{\epsilon}$.
Also note that for the total angular momentum of the particles: 1) the grid node locations can be perceived as the sample points of a rotating mass centered at the material particle location, and 2) the total angular momentum comprises both that of the center and that of the affine-rotation of the grids around the center.

APIC P2G is given by
\begin{align}
    \label{apicp2g}
    \begin{split}
        m_i^n &= \sum_p w_{p, i}^n m_p \\
        \mD_p^n &= \sum_i w_{p, i}^n (\vx_i^n-\vx_p^n) (\vx_i^n-\vx_p^n)^T\\
        m_i^n \vv_i^n &= \sum_p w_{p, i}^n m_p (\vv_p^n + \mB_p^n (\mD_p^n)^{-1} (\vx_i^n-\vx_p^n))
    \end{split}
\end{align}
with G2P given by
\begin{align}
    \label{apicg2p}
    \begin{split}
        \vv_p^{n+1} &= \sum_i w_{p, i}^n\tilde{v}_i^{n+1} \\
        \mB_p^{n+1} &= \sum_i w_{p, i}^n \tilde{v}_i^{n+1} (\vx_i^n-\vx_p^n)^T,
    \end{split}
\end{align}
where the superscript $\tilde{}$ means the intermediate value after the update on grids but before the G2P process.

\subsection{Numerical Validation}
\label{appdx:apic:conservation:experiment}

\begin{figure}[tb]
    \centering
    \includegraphics[width=0.9\textwidth]{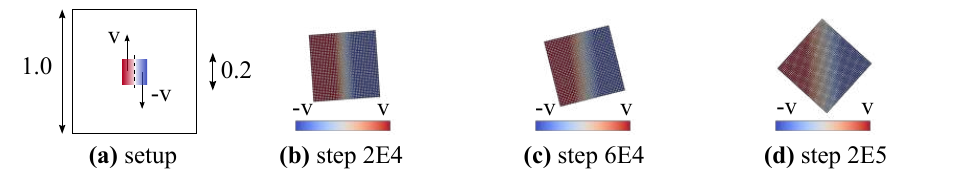}
    \caption{Setup and snapshots of a rotating elastic square.}
    \label{fig:2d_rotating_cube_setup}
\end{figure}

\begin{figure}[tb]
    \centering
    \includegraphics[width=0.9\textwidth]{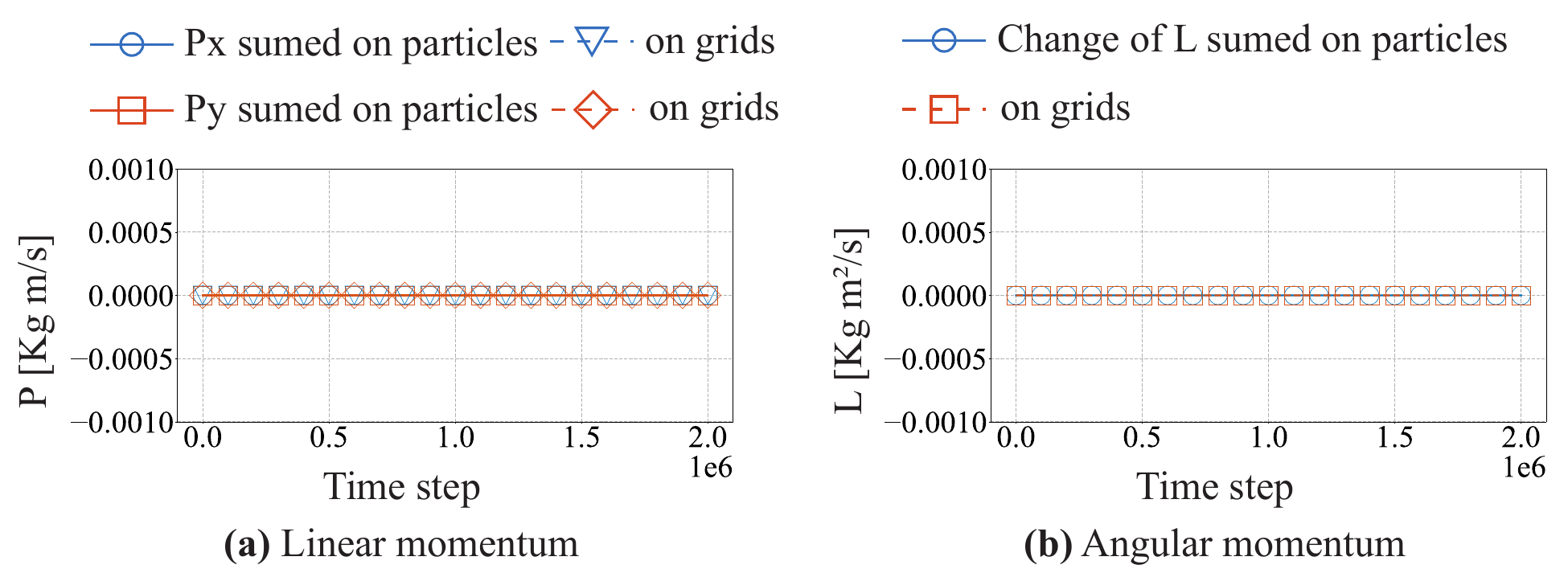}
    \caption{Logs of (a) linear and (b) angular momentum of the rotating cube experiment after $10^6$ time steps.}
    \label{fig:2d_rotating_cube_conservation}
\end{figure}

A numerical validation as in \cite{jiang2015affine} is also conducted here to verify these conservations. A square with a side length of $l=0.2$ is discretized with $20 \times 20$ particles. The physical properties of the square are as follows: $E=1\times10^4$\,Pa, $\nu = 0.3$, and $\rho=1.0$\,kg/m$^3$. Initially, the square is divided into two halves by a hypothetical vertical line through the middle. The left half is initialized with an upward velocity $\mathbf{v}= (1,0)$\,m/s, while the right half is initialized with a downward velocity $\mathbf{v}= (-1,0)$\,m/s.
The experimental setup is illustrated in \Figref{fig:2d_rotating_cube_setup}.
The background mesh is generated using Delaunay triangulation with a target element size of $0.01$\,m in a $1 \times 1$\,m$^2$ box. The simulation is run for $1\times10^6$ time steps with a time step size of $1\times10^{-5}$\,s.

The proposed conservation is accurately illustrated in Figure~\ref{fig:2d_rotating_cube_conservation}b--c, with only round-off errors on the order of $1 \times 10^{-15}$ and $1 \times 10^{-7}$ for the total linear and affine momentum of the system, respectively.

\end{document}